\documentclass[11pt]{article}
\usepackage{amsmath, amsthm}
\usepackage{amsmath}
\usepackage{amsfonts}
\usepackage{amssymb}
\usepackage{verbatim}
\usepackage{epsfig}
\usepackage{times}
\usepackage{mathrsfs}
\usepackage{color}
\definecolor{purple}{rgb}{0.65, 0, 1}
\definecolor{orange}{rgb}{1,.5,0}
\definecolor{juice}{rgb}{.9,.43,.26}

\usepackage{color} 

\definecolor{purple}{rgb}{0.65, 0, 1}
\definecolor{orange}{rgb}{1,.5,0}

\def\t{_{\text{\tiny T}}}

\newtheorem{theorem}{Theorem}[section]
\newtheorem{lemma}[theorem]{Lemma}

\newtheorem{proposition}[theorem]{Proposition}
\newtheorem{corollary}[theorem]{Corollary}

\begin{document}

\title
{\Large{The McKean\hspace{2 pt}--Vlasov Equation in Finite Volume}}

\author
{\large{ L. Chayes$^{1}$ and V. Panferov$^{2}$}}
\maketitle

\vspace{-4mm}

\centerline{$^1$\textit{Department of Mathematics, UCLA, Los Angeles,
CA 90059--1555 USA}}

\centerline{$^2$\textit{Department of Mathematics, 
CSUN, Northridge, CA 91330--8313 USA}}

\begin{quote}
{\footnotesize {\bf Abstract:}
We study the McKean\hspace{1 pt}--Vlasov equation on the finite tori of length scale $L$ in $d$--dimensions.  
We derive the necessary and sufficient conditions for the existence of a phase transition, which are based 
on the criteria first uncovered in \cite{GP} and \cite{KM}. Therein and in subsequent works, one finds 
indications pointing to critical transitions at a particular model dependent value, $\theta^{\sharp}$ of 
the interaction parameter. We show that the uniform density (which may be interpreted as the liquid phase) 
is dynamically stable for $\theta < \theta^{\sharp}$ and prove, abstractly, that a {\it critical} transition 
must occur at $\theta = \theta^{\sharp}$.  However for this system we show that under generic 
conditions -- $L$ large, $d \geq 2$ and isotropic interactions -- the phase transition is in fact 
discontinuous and occurs at some $\theta\t < \theta^{\sharp}$. Finally, for H--stable, bounded interactions 
with discontinuous transitions we show that, with suitable scaling, the $\theta\t(L)$ tend to a definitive 
non--trivial limit as $L\to\infty$.

}
\end{quote}
\section{Introduction}  This paper concerns the McKean\hspace{1.5pt}--Vlasov equation -- so called in \cite{V} -- which is a non--linear diffusion equation the classical rendition of which reads 
\begin{equation}
\label{MV}
\rho_{t}  = \Delta \rho + \theta L^{d}\hspace{.5 pt} \nabla \cdot \rho \nabla (V\star \rho).
\end{equation}
In the above $\rho = \rho(x,t)$, we take 
$x \in \mathbb T_{L}^{d}$ -- the $d$--dimensional torus of scale $L$ -- and $\star$ denotes convolution.  
It is noted that the above dynamics is positivity and $L^{1}$--norm preserving thus $\rho(x,t)$ has a probabilistic 
interpretation which we relate to {\it particle density}.  It is hereafter assumed that $\rho$ integrates to one.  
As is well known, the dynamics in Eq.(\ref{MV}) is governed by gradient flow for the (``free energy'') functional
\begin{equation}
\label{FE}
\mathcal F_{\theta}(\rho)  =  \int_{\mathbb T_{L}^{d}}\rho\log\rho dx+ 
\frac{1}{2}\theta L^{d}\int_{\mathbb T_{L}^{d}\times \mathbb T_{L}^{d}}
\hspace{- 10 pt}
V(x-y)\rho(x)\rho(y)dxdy.
\end{equation}
In particular all steady state solutions of 
Eq.(\ref{MV}) must be stationary points of the functional in
Eq.(\ref{FE}).  
These densities satisfy an Euler--Lagrange equation, namely
\begin{equation}
\label{VEV}
\rho(x)  =  \frac{\text{e}^{-\theta L^{d}[\rho\star V](x)}}
{\int_{\mathbb T_{L}^{d}} \text{e}^{-\theta L^{d}\rho\star V}dx}
\end{equation}
sometimes known as the Kirkwood--Monroe equation \cite{KM}.  
(The above follows from the fact, readily checked, that the dynamical equation can be recast into the form
$$
\frac{d}{dt}\mathcal F_{\theta}(\rho)  =
-\int_{\mathbb T_{L}^{d}}\rho
\left |
\nabla \log \frac{\rho}{\text{e}^{-\theta L^{d}V\star\rho}}
\right |^{2}dx.)
$$
The volume factor, $L^{d}$, associated with the {\it coupling strength} in Eqs.(\ref{MV})--(\ref{VEV}) 
may appear unfamiliar to some but it is in fact a principal subject of this note.   

We shall not digress with a detailed discussion of the motivations for the study of Eqs.(\ref{MV})--(\ref{FE}).  It is sufficient to mention the following:

$\bullet$  Eq.(\ref{MV}) can be realized as the large $N$--limit of the $N$--particle Fokker--Planck equation under suitable rescaling of the interaction. 
This goes back to the original derivation by McKean \cite{McK}
and, even today, is an active topic of mathematical research.  A partial list of relevant papers:
\cite{BL2}, \cite{JT2},
\cite{Szn},\cite{aS}, \cite{BGM},
\cite{Mal}.

$\bullet$  Eq.(\ref{MV}) can be realized as a diffusive limit of the standard
Vlasov--Fokker Plank equation. 
Cf.~the derivation in \cite{MaMa}.

$\bullet$  The model of chemotaxis introduced by 
Keller and Segel \cite{KS} is, in fact, precisely the
McK--V equation in slightly disguised form with a Newtonian (logarithmic) interaction; cf.~\cite{aS}, \cite{HS} for a derivation from particle dynamics. 
For our purposes, the Keller--Segel form of the interaction is overly singular -- by no means a requirement dictated by biological applications.  Related models with biological applications are described in \cite{LR},
\cite{DCBC}, \cite{CDMBC}, \cite{CDP},
\cite{aBjCtL}, \cite{aBtL}.  The latter two are {\it exactly}
the McK--V equation without the diffusive term.

%This is currently a model which is the subject of intensive (and interesting) study but not particularly relevant for the current work

%Biological applications.  \noindent\textcolor{blue}{Mention (a)  Keller Segal, Swarming work,  Northridge work on subject and Bertozzi \& Laurent etc. work without the diffusion term.} \cite{HS}

$\bullet$ In a number of older works, beginning with \cite{Kac} and
\cite{KUH} and including (but not limited to) 
\cite{vKa}
\cite{Ga}, \cite{GP}, and \cite{LP}  
(cf.~the article \cite{CjT} for additional information and references) the van der Waals theory of interacting fluids in statistical equilibrium was elucidated as the limit of ``realistic'' systems under {\it scaling} of the interaction range.  A modified version of the functional in Eq.(\ref{VEV}), evaluated at its minimizer constitutes the free energy for these (limiting) theories.  
Finally, in the remarkable work \cite{KM} -- predating all of the above by over two decades -- the equation (\ref{VEV}) for the equilibrium ``distribution function'' was inferred, under certain approximations, by direct considerations.  

\vspace{.125 cm}

It should be remarked that the scaling limits achieved in the first item are not always in accord with those of the last.  As such the volume factor is conspicuously absent in many modern mathematical treatments of these and related problems.  However, on careful examination, the latter derivations contain the former in the static cases.  Thus, for physically motivated {\it stable} interactions (which will be discussed in Section 4) with sensible thermodynamic  limits, this factor indeed belongs as written in Eqs.(\ref{MV}) -- (\ref{VEV}).  
For unstable interactions -- which may have biological applications -- the correct nature of the scaling has not been elucidated.  However it appears that mathematically tractable problems in large or infinite volume emerge if the factor of $L^{d}$ is omitted.

\subsection{Mathematical Assumptions and Notations} 
Since the majority of this work takes place in fixed volume, we will omit, whenever possible, the $L$--dependence in 
our notation for the various classes of functions etc. that we employ.  In particular all $L^{p}$--norms on 
$\mathbb T_{L}^{d}$ will be unadorned.

The class of potentials that we consider in this work are described as follows:  Foremost we shall assume that 
the $V$ are finite range, that is
$$
V(x) = 0 \text{ if } |x| > a.
$$
We will always take $L > a$ and thus we may define the remaining (minimalist) properties as though $V$ is a function on $\mathbb R^{d}$.  First, we take
$V\in L^{1}$ and, second we assume that $V$ is bounded below.  The former is obviously required in order to make (good) sense of the uniform state.  As for the latter, if $V\to-\infty$
it is unreasonable to suppose that this happens anywhere besides the origin.  Even mild divergence (e.g., logarithmic in $d = 2$) can cause the functional to be unbounded below (and, in fact, just having $V < 0$ a.e.~in a neighborhood of the origin leads to unphysical behavior).  Finally, on physical grounds, we shall assume that $V$ is a symmetric function of its argument:
$V(x) = V(-x)$.
We shall denote the class by $\mathscr V$:
\begin{equation}
\label{V}
\mathscr V = \{V\in L^{1} \text{ s.t. } V^{-}
\hspace{-3 pt}
\in L^{\infty}
\text{ and } V 
\text{ symmetric with } V(x) = 0 \text{ for } |x| > a \}
\end{equation}
where $V^{-}$ denotes the negative part of $V$ and $a < L$.   Additional technical assumptions will be implemented as needed.  

For the analysis of the functional $\mathcal F_{\theta}$, we shall denote by $\mathscr P$ the class of probability densities on $\mathbb T_{L}^{d}$ (although it is clear that $\mathscr P$ is much larger than necessary).  The uniform density, will be denoted by $\rho_{0}$:
\begin{equation}
%\label{}
\rho_{0} := L^{-d}.
\end{equation}

We denote the separate pieces of $\mathcal F_{\theta}(\cdot)$ by an $\mathcal S$ and 
$\mathcal  E$:  $\mathcal F_{\theta}(\cdot) =: \mathcal S(\cdot) + 
\frac{1}{2}\theta \rho_{0}^{-1}\mathcal E(\cdot, \cdot)$. For 
the second, we will often have occasion to regard as this as the bilinear functional
$$
\mathcal E(\rho_{a},\rho_{b}) 
: =
\int_{\mathbb T_{L}^{d}\times \mathbb T_{L}^{d}}
\hspace{- 10 pt}
V(x-y)\rho_{a}(x)\rho_{b}(y)dxdy 
$$
which is (usually) defined regardless of the signs or normalization of its arguments.  
Likewise, we will have some occasion to utilize the functional $\mathcal S$
for arguments that, albeit non--negative, may not be  normalized.
For a legitimate non--negative, normalized $\rho(x)$, the
quantities $\mathcal S(\rho)$ and $\frac{1}{2}\theta L^{d}\mathcal E(\rho,\rho)$
are (modulo signs) vaguely related to the entropy 
and energy of the system when the 
equilibrium density is $\rho$; the two terms will indicated by these names.

\subsection{Summary and Statement of Results.}

The central purpose of this note is the study of these systems as $\theta$ varies.  Often enough these systems go from a quiescent (gaseous) state where no minimizers of $\mathcal F_{\theta}(\cdot)$ exist save for the uniform state to a state where this is no longer the minimizer and other minimizers are prevalent.  
In short, a phase transition the nature of which we shall partially elucidate.  The results proved and their relevant location are as follows:

In Section 2, the subject of phase transitions in the McK--V system will be discussed from the ground up.  First, in $\mathsection$ 2.1 (which may be omitted on a preliminary reading) we establish the existence of minimizers.  This allows, in $\mathsection$ 2.2 a
``thermodynamic'' definition of the entropic and energetic content of the system as a function of the interaction parameter 
$\theta$ which in turn will clarify the definition and possible nature of the (lower) transition point.  In $\mathsection$ 2.3, necessary and sufficient conditions (on $V$) are established for the occurrence of a phase transition.  The candidate transition point, much discussed in other works and here denoted by 
$\theta^{\sharp}$ is elucidated and it is shown that for 
$\theta < \theta^{\sharp}$, the uniform density is dynamically stable.  In $\mathsection$ 2.4, a concise definition of a (lower) 
{\it critical}  transition point is provided.  First it is demonstrated (under the additional and presumably unnecessary assumption that $V \in L^{2}$) that if such a transition occurs, it must take place at $\theta = \theta^{\sharp}$.  Then it is shown that the features of a non--critical transition (where the above mentioned criteria fail) are dramatically different.  The subsection ends with a principal result of this note.  Namely under the majority of physically -- or for that matter biologically -- reasonable circumstances, it is a 
{\it non--critical} transition which occurs in the McK--V system.  Moreover these occur at parameter value $\theta\t$ which is {\it strictly}
smaller than $\theta^{\sharp}$.  
Finally in Section 3, the limiting behavior in large volume is discussed.  In 
$\mathsection$3.1 it is shown that, for fixed interaction, the $L\to\infty$ limit of the transition points always exist.  But the limit may be trivial.  In $\mathsection$ 3.2, a criterion closely related to H--stability is introduced and it is shown that (with the scalings featured in Eqs.(\ref{MV})--(\ref{VEV})) for stable potentials the transition points tend to a definitive non--trivial limit.  Conversely,
in $\mathsection$ 3.3, the 
complimentary -- catastrophic -- cases, 
are investigated and it is shown that
the transition values tend quickly to zero.

%

%\vspace{3 cm}

%\noindent (o)  Always working on $\mathbb T_{L}^{d}$.  May as well assume that $V \in L^{2}$

%\noindent (a)  Uniform density, $\rho_{0}$, is always a solution.

%\noindent (b)  When $\theta$ sufficiently small, $\rho_{0}$ is the only solution.  When $V$ is of positive type, 
%$\rho_{0}$ is the only solution.  

%\noindent (c)  When $V$ is {\it not} of positive type, 
%$\exists \theta^{\#}$ above which (i)  linearly unstable (ii) functional no longer maximized by constant solution and 
%there are non--constant solutions which are better candidates.

%\noindent (d)  First order transitions.

%\noindent (e)  Catastrophes in the catastrophic cases (absence of scaling properties).

%\vspace {5 pt}

\section{Phase Transitions}
\subsection{Minimizing solutions}
The starting point in our analysis is to establish, for all $\theta$, the existence of stationary solutions to Eq.(\ref{MV}) that minimize the free energy functional 
in Eq.(\ref{FE}).  The existence of minimizers for these sorts of problems has a history: In particular 
\cite{BL} discuss the existence of minimizers for functionals of this form referring back to the works
\cite{GP}, \cite{Ga}. In \cite{CCELM} there is an explicit construction for a related problem and, recently \cite{Co},  established the desired result by methods not dissimilar to those presently employed.  
We shall include a proof for completeness which is succinct given the following:

\begin{lemma}
\label{BbbB}
Let $\mathcal F_{\theta}(\cdot)$ be as described.  
Then $\exists B_{0}
< \infty$
such that for all $\rho \in \mathscr P$. 
the following holds: if
$||\rho||_{\infty} > B_{0}$ then there is another
$\rho^{\ddag} \in \mathscr P$ for which
$$
\mathcal F_{\theta}(\rho)
>
\mathcal F_{\theta}(\rho^{\ddag}).
$$
\end{lemma}
\begin{proof}
We start with the observation that for any 
$\rho$,
$$
\mathcal S(\rho) \geq \mathcal S(\rho_{0})
= -\log L^{d}
$$
and
$$
\frac{1}{2}\theta L^{d}\mathcal E(\rho,\rho)
\geq
-\frac{1}{2}\theta L^{d} V_{0}
$$
where $-V_{0}$ is the lower bound on $V(x)$.

For $\rho \in \mathscr P$, and $B > 0$ let
$\mathbb B_{B}(\rho)$ denote the set
$$
\mathbb B_{B}(\rho) =
\{x\in \mathbb T_{L}^{d} \mid \rho \geq B\}
$$
and $\varepsilon_{B}(\rho)$ the $\rho$--measure of
$\mathbb B_{B}$:
$$
\varepsilon_{B}  =  
\int_{\mathbb B_{B}}\rho dx.
$$
We shall (rather arbitrarily) divide into the two cases
of $(\rho,B)$'s for which
$\varepsilon_{B}(\rho) \geq \frac{1}{2}$
and 
$\varepsilon_{B}(\rho) < \frac{1}{2}$

Obviously if $\varepsilon_{B}(\rho) \geq \frac{1}{2}$
then
$$
\mathcal S(\rho) \geq \frac{1}{2}\log B
+ \int_{\mathbb B_{B}^{c}} 
\hspace{-3pt}
\rho \log \rho.
$$
The second term may be estimated from below by
$(1 - \varepsilon_{B})\log(1 - \varepsilon_{B})
+(1 - \varepsilon_{B})\log|\mathbb B_{B}^{c}|$ which can be bounded by quantities which do not depend on $B$.
Since the energy term is bounded below it is clear 
that for some $B_{1} < \infty$ if
$B > B_{1}$
$\mathcal F_{\theta}(\rho)$ will exceed 
$\mathcal F_{\theta}(\rho_{0})$.  We turn attention to the cases $\varepsilon_{B} < \frac{1}{2}$.

We write $\rho = \rho_{b} + \rho_{r}$ where
$\rho_{b}$ is the restriction of $\rho$ to the 
set $\mathbb B_{B}$ and 
$\rho_{r}$ is the rest.
Our claim is that if $B$ is too large then
$$
\mathcal F_{\theta}(\rho) > 
\mathcal F_{\theta}(\overline{\rho}_{r})
$$
were
$\overline{\rho}_{r} =
(1 - \varepsilon_{B})^{-1}\rho_{r}$ is the
normalized version of 
$\rho_{r}$.

We write, 
$\mathcal S(\rho_{r}) := \int_{\mathbb T_{L}^{d}} \rho_{r}\log \rho_{r}dx$, notwithstanding the fact that 
$\rho_{r}$ is not normalized,
and observe (assuming $B > 1$) that
$$
\mathcal S(\rho) = {\mathcal S}(\rho_{r}) + \varepsilon_{B}\log B >  {\mathcal S}(\rho_{r}).
$$
Since we might as well assume that 
$\mathcal F_{\theta}(\rho) \leq \mathcal F_{\theta}(\rho_{0}$) and the energetic components of both of these quantities are bounded above and below this implies that for some $s_{\star} < \infty$ 
$$
s_{\star} \geq \mathcal S(\rho) > {\mathcal S}(\rho_{r})
$$
regardless of the particulars of $\rho$ vis--$\grave{\text{a}}$--vis $B$ and 
$\varepsilon_{B}$.   Similarly, we have (since
$\mathcal E(\rho, \rho) < \mathcal E(\rho_{0}, \rho_{0})$ and $\varepsilon_{B} < 1/2$) that
$\mathcal E(\rho_{r}, \rho_{r}) \leq \mathcal E(\rho_{0}, \rho_{0}) + V_{0}
= : e^{\star} < \infty$.

Now let us estimate $\mathcal F_{\theta}(\rho) - 
\mathcal F_{\theta}(\overline{\rho}_{r})$.  
First
%\begin{align}
%\mathcal S(\rho) - \mathcal S(\overline{\rho}_{r}) 
%&\geq \varepsilon_{B}\log B -
%\frac{\varepsilon_{B}}{1 - \varepsilon_{B}}{\bf S}(\rho_{r}) +
% \log (1 - \varepsilon_{B})
% \\
%&\geq \varepsilon_{B}\left[
% \log B -\frac{1 + s^{\star}}{1 - \varepsilon_{B}}
% \right]
%\end{align}
\begin{equation}
\label{BOH}
\begin{split}
\mathcal S(\rho) - \mathcal S(\overline{\rho}_{r}) 
&\geq \varepsilon_{B}\log B -
\frac{\varepsilon_{B}}{1 - \varepsilon_{B}}{\mathcal S}(\rho_{r}) +
 \log (1 - \varepsilon_{B})
 \\
&\geq \varepsilon_{B}\left[
 \log B -\frac{1 + s^{\star}}{1 - \varepsilon_{B}}
 \right]
\end{split}
\end{equation}
where we have used the fact that $ \log (1 - \varepsilon_{B}) \geq 
-\frac{\varepsilon_{B}}{1 - \varepsilon_{B}}$.  
As for the energetics, it is seen that
$$
\mathcal E(\rho,\rho)  
\geq \mathcal E(\rho_{p},\rho_{p}) -\theta V_{0}\varepsilon_{B}
$$
while 
$$
%\label{}
\mathcal E(\overline\rho_{p}, \overline\rho_{p})  =  
\frac{1}{(1 - \varepsilon_{B})^{2}} \mathcal E(\rho_{p}, \rho_{p}).
$$
so 
\begin{equation}
\label{KOH}
 \mathcal E(\rho, \rho)
- \mathcal E(\overline\rho_{r}, \overline\rho_{r})
\geq 
[\frac{-2\varepsilon_{B}+\varepsilon_{B}^{2}}{(1 - \varepsilon_{B})^{2}}]
\mathcal E(\rho_{r},\rho_{r})
\geq -8 \varepsilon_{B} e_{\star}
\end{equation}
where we have used $\varepsilon_{B} < 1/2$.  

The combination of Eq.(\ref{BOH}) and Eq.(\ref{KOH}) show that if $B$ exceeds
some (finite) $B_{2}$ the density $\overline\rho_{r}$ represents an ``improvement''.  Note that, conceivably, the improvement may take values as large as {\it twice} $B_{2}$.  Nevertheless,  
the theorem is completed by declaring 
$B_{0} = \max\{B_{1}, 2B_{2}, 1\}$ and using for 
$\rho^{\ddag}$
the uniform 
or above described density as appropriate.
\end{proof}
\begin{theorem}
\label{MMM}
Let $\mathcal  F_{\theta}(\rho)$ be as described in Eq.(\ref{FE})
Then there exists a 
$\rho_{\theta} \geq 0 \in \mathscr P$
that minimizes $\mathcal F_{\theta}(\cdot)$.
\end{theorem}
\begin{proof}
Let $(\rho_{j})$ denote a minimizing sequence for 
$\mathcal F_{\theta}(\cdot)$.  Since, without loss of generality
$\rho_{j} \in L^{1}\cap L^{\infty}$, we may place
the $\rho_{j}$ in $L^{2}$ with $||\rho_{j}||_{2}^{2} < B_{0}$.
Let $\rho_{\infty}$ denote a weak limit of the sequence.  By standard convexity arguments, 
$\lim_{j\to\infty} \mathcal S(\rho_{j}) \geq \mathcal S(\rho_{\infty})$
(where we have used $(\rho_{j})$ to denote the {\it sub}sequence).

We claim that
\begin{equation}
%\label{}
\lim_{j\to\infty}\mathcal E(\rho_{j}, \rho_{j})  =
\mathcal E(\rho_{\infty}, \rho_{\infty}).
\end{equation}
This follows from some elementary Fourier analysis:  
Since $V(x) \in L^{1}$, $|\hat{V}(k)| \to 0$ as
$k \to \infty$.  Let $\overline{w}_{k_{0}} = \max_{|k| > |k_{0}|}|\hat V(k)|$
where here and throughout it is assumed that 
all $k$'s are legitimate wave vectors for $\mathbb T_{L}^{d}$.  
Then 
\begin{equation}
%\label{}
|\mathcal E(\rho_{j},\rho_{j}) 
- \mathcal E(\rho_{\infty},\rho_{\infty})|
\leq |
\hspace{-6 pt}
\sum_{k: |k| \leq |k_{0}|}
\hspace{-6 pt}
\hat{V}(k)[|\hat{\rho}_{j}(k)|^{2} - |\hat{\rho}_{\infty}(k)|^{2}]
+ 2B_{0}\overline{w}_{k_{0}}.
\end{equation}
The first term tends to zero since for each individual $k$, 
$\hat{\rho}_{j}(k) \to \hat{\rho}_{\infty}(k)$
and the second term can be made as small as desired.
Thus we may conclude that $\rho_{\infty}$ actually minimizes the functional.
\end{proof}

%\noindent (a)  Prove fact that $\rho(x) > 0$

On the basis of the above, we may define
\begin{equation}
\label{MXM}
\mathscr M_{\theta}  :=  
\{\rho \in \mathscr P\mid \rho \text{ minimizes } 
\mathcal F_{\theta}(\cdot)\}
\end{equation}
with the assurance that $\forall \theta$, 
$\mathscr M_{\theta} \neq \emptyset$.  As an obvious corollary to Lemma \ref{BbbB}, we have that any $\rho \in \mathscr M_{\theta}$ is bounded above.  
Conversely, we have uniform lower bounds (which, strictly speaking, do not play a r\^{o}le in later developments).

\begin{proposition}
Let $\rho \in \mathscr M_{\theta}$.  Then $\rho$ is bounded below strictly away from zero.
\end{proposition}
\begin{proof}
We appeal directly to the Kirkwood--Monroe equation (Eq.(\ref{VEV})) from which it is clear that pointwise upper and lower bounds on 
$V\star\rho$ are sufficient.  Obviously
$V\star\rho \leq ||\rho||_{\infty}||V||_{1}$.  Next, with more elaboration than may be necessary, let
\begin{equation}
%\label{}
P_{a}(\rho)  =  \sup_{y\in\mathbb T_{L}^{d}}
\int_{|y-y^{\prime}| \leq a}
\hspace{-6 pt}
\rho(y^{\prime})dy^{\prime}
\end{equation}
where it is recalled that $a$ denotes the range of the interaction.  Then $V\star\rho \geq -P_{a}V_{0}$.  This provides
$$
\rho(x)  \geq  \exp-\theta L^{d}
[P_{a}V_{0} +  ||\rho||_{\infty}||V||_{1}] > 0.
$$
\end{proof}

\noindent {\bf Remark} \hspace {2 pt}
It is anticipated 
that in physically reasonable (stable) cases, which will be discussed in Section 4, both terms in the square bracket appearing in the previous equation are of the order of $L^{-d}$.  However, in catastrophic cases, it seems that 
$P_{a}(\rho)$ will indeed achieve values of order unity independent of $L$ for 
$\rho \in \mathscr M_{\theta}$.  

%\noindent (b) :
%\label{BBB}
%\noindent\textcolor{magenta}{\begin{corollary} Corollary to theorem is that any minimizer is bounded.  Use arguments that were used on minimizing sequence to ``improve'' minimizer.  
%\end{corollary}}

%
%\noindent (c)  Fact that $\rho_{0}$ always a linearly stable solution in the sense that 
%\begin{equation}
%\label{CDC}
%\left .
%\frac{\delta \mathscr F}{\delta \rho}\right|_{\rho_{0}}  =  0.
%\end{equation}
%Must always investigate stability to quadratic order.

\subsection{Thermodynamics for the McK--V System}

We may now separately define the energetic and entropic content of the system as a function of the parameter $\theta$; these form the basis of a thermodynamic theory.

\noindent {\bf Definition} \hspace {2 pt}
We define
\begin{equation}
% \label{}
E_\theta  =  \inf_{\rho \in \mathscr M_{\theta}} \frac{1}{2}\theta\rho_{0}^{-1}
\mathcal E(\rho,\rho)
\end{equation}
and 
\begin{equation}
%\label{}
S_\theta  =  \inf_{\rho \in \mathscr M_{\theta}} 
\mathcal S(\rho)
\end{equation}
Furthermore, defining 
\begin{equation}
%\label{}
F_{\theta}  =  \inf_{\rho \in L^{1}(\mathbb T_{L}^{d})} \mathcal F(\rho)
\end{equation}
we have, to within signs and constants, the energy, entropy and free energy of the system at parameter value 
$\theta$.  It is noted that the first two do not always add up to the third.  
\begin{proposition}
\label{KMO}
Consider the above defined thermodynamic functions.  Then\\
(a)  $S_{\theta}$ is non--decreacing\\
(b)  $F_{\theta} -\frac{1}{2}\theta\rho_{0}^{-1}\mathcal E(\rho_{0},\rho_{0})$  is non--increacing and continuous\\
while\\
(c)  $\theta^{-1}E_{\theta}$ and 
$E_{\theta}-\frac{1}{2}\theta\rho_{0}^{-1}\mathcal E(\rho_{0},\rho_{0})$ are non--increacing.
\end{proposition}

We remark that the subtractions are actually necessary: consider, e.g.,~the situation -- which is the general rule for ``reasonable'' systems -- where $\mathcal E(\rho_{0},\rho_{0}) > 0$ and $\mathcal F(\cdot)$ is always minimized 
by $\rho_{0}$ for all values of $\theta$ that are sufficiently small.

\begin{proof}
We shall start with the energetics.  Let $\theta_{1} , \theta_{2} \geq 0$ and let 
$\rho_{\theta_{1}} \in \mathscr M_{\theta_{1}}$ and similarly for $\rho_{\theta_{2}}$.  Then, using $\rho_{\theta_{2}}$ instead of $\rho_{\theta_{1}}$ we have that
\begin{align*}
%\label{}
F_{\theta_{1}}  \leq  \mathcal F_{\theta_{1}}(\rho_{\theta_{2}})
&=  \mathcal F_{\theta_{2}}(\rho_{\theta_{2}})
- \frac{1}{2}\rho_{0}^{-1}(\theta_{2} - \theta_{1})\mathcal E(\rho_{\theta_{2}},\rho_{\theta_{2}})\\
&\leq  F_{\theta_{2}} -  \frac{1}{2}(\theta_{2} - \theta_{1})\rho_{0}^{-1}\mathcal E(\rho_{\theta_{2}},\rho_{\theta_{2}}).
\end{align*}
Similarly, 
$$
F_{\theta_{2}}  \leq  F_{\theta_{1}}  -  
\frac{1}{2}(\theta_{1} - \theta_{2})\rho_{0}^{-1}\mathcal E(\rho_{\theta_{1}},\rho_{\theta_{1}})
$$
so that 
$(\theta_{2} - \theta_{1})\mathcal E(\rho_{\theta_{2}},\rho_{\theta_{2}})
\leq  (\theta_{2} - \theta_{1})\mathcal E(\rho_{\theta_{1}},\rho_{\theta_{1}})$
which, if $\theta_{2} > \theta_{1}$, certainly implies the first of the items in (c).  
However, a bit more has been shown:  The energetic content of 
{\it any} $\rho_{\theta_{1}} \in \mathscr M_{\theta_{1}}$ is monotonically 
related to the energetic content of {\it any}
$\rho_{\theta_{2}} \in \mathscr M_{\theta_{2}}$.  

This immediately establishes monotonicity of the entropy--term.  Indeed, suppose that, $\theta_{2} > \theta_{1}$.
Then, at $\theta = \theta_{1}$ using a $\rho_{\theta_{2}}$ we have:
\begin{equation}
%\label{}
F_{\theta_{1}} \leq \mathcal S(\rho_{\theta_{2}}) + 
\frac{1}{2}\theta_{1}\rho_{0}^{-1}\mathcal E(\rho_{\theta_{2}},\rho_{\theta_{2}}).
\end{equation}
The energy term is less than that associated with
$\mathcal E(\rho_{\theta_{1}},\rho_{\theta_{1}})$
and we arrive at
\begin{equation}
%\label{}
F_{\theta_{1}}  \leq  \mathcal S(\rho_{\theta_{2}}) + 
\frac{1}{2}\theta_{1}\rho_{0}^{-1}
\mathcal E(\rho_{\theta_{1}},\rho_{\theta_{1}})
=  F_{\theta_{1}} + S(\rho_{\theta_{2}}) - S(\rho_{\theta_{1}}).
\end{equation}
We again have that for {\it any} $\rho_{\theta_{1}} \in \mathscr M_{\theta_{1}}$
and  $\rho_{\theta_{2}} \in \mathscr M_{\theta_{2}}$, with $\theta_{1} < \theta_{2}$,
$$
\mathcal S(\rho_{\theta_{1}})  \leq  \mathcal S(\rho_{\theta_{2}}).
$$

As for the claims about $F_{\theta}$, continuity follows from the first two displays in this proof.   For the monotonicity, of 
$E_{\theta}-\frac{1}{2}\theta\rho_{0}^{-1}\mathcal E(\rho_{0},\rho_{0})$,
we first observe that for any $\theta$ and any 
$\rho_{\theta} \in \mathscr M_{\theta}$
$$
\mathcal E(\rho_{\theta}, \rho_{\theta})  \leq \mathcal E(\rho_{0}, \rho_{0}).
$$
with equality only if $\rho_{\theta}  =  \rho_{0}$ a.e.
Indeed, assuming that $\rho_{\theta}$ is not a.e.~equal to $\rho_{0}$, then
$\mathcal S(\rho_{0}) < \mathcal S(\rho_{\theta})$ so $\rho_{\theta}$ could not possibly be a minimizer if the opposite of the above display were to hold.  Then 
$[\mathcal E(\rho_{\theta},\rho_{\theta}) - \mathcal E(\rho_{0},\rho_{0})]$
is non--positive and non--increasing so 
$\theta[\mathcal E(\rho_{\theta},\rho_{\theta}) - \mathcal E(\rho_{0},\rho_{0})]$
is non--increasing.

The final claim is now proved by reiteration of the previous procedures with the subtraction in place:
\begin{align*}
% \label{}
F_{\theta_{2}} -&\frac{1}{2}\theta_{2}\rho_{0}^{-1}\mathcal E(\rho_{0}, \rho_{0})
\leq  \mathcal S(\rho_{\theta_{1}})    +\frac{1}{2}\theta_{2}\rho_{0}^{-1}\mathcal 
E(\rho_{\theta_{1}},\rho_{\theta_{1}})
-\frac{1}{2}\theta_{2}\rho_{0}^{-1}\mathcal E(\rho_{0}, \rho_{0})\\
& =  F_{\theta_{1}} -\frac{1}{2}\theta_{1}\rho_{0}^{-1}\mathcal E(\rho_{0}, \rho_{0})
+\frac{1}{2}(\theta_{2} -
 \theta_{1})\rho_{0}^{-1}[\mathcal E(\rho_{\theta_{1}},\rho_{\theta_{1}}) - \mathcal E(\rho_{0}, \rho_{0})]
\end{align*}
where we have assumed $\theta_{2} > \theta_{1}$.  By the non--positivity of the quantity in the square brackets, the stated monotonicity is established.
\end{proof}

With the above monotononicities in hand, the objects 
$S_{\theta}$ and $E_{\theta}$ can now be considered well defined functions for all $\theta$ which are continuous for a.e.~$\theta$.   However at points of discontinuity, it may be more useful to focus on the range of the function rather than its value at the point.  In particular $S_{\theta}$ is continuous iff $E_{\theta}$ is continuous while at points of discontinuity, the density that minimizes $\mathcal S$ is the one that maximizes $\mathcal E$ and vice versa.
% Take sequence of higher/lower valued maximizers at points of continuity with the $\theta$ converging to desired value.

\subsection{Phase transitions in the McK--V systems (1): \\ The point of linear stability}

We start this subsection with some preliminary results 
-- most of which have appeared elsewhere in the literature 
(albeit by different methods) -- 
concerning the single phase regime: The regime where 
$\rho_{0}$ is the unique minimizer of 
$\mathcal F_{\theta}(\cdot)$.
\begin{proposition} 
\label{prop:conv}
Let $V \in \mathscr V$ be bounded i.e. $|V| \leq V_{\rm{max}}
< \infty$.  Then for 
$\theta L^{d}$ sufficiently small 
-- less than $[V_{\rm{max}}]^{-1}$ --
the functional $\mathcal F_{\theta}(\rho)$ is convex. 
\end{proposition}
\begin{proof}
The functional $\mathcal F_{\theta}(\rho)$ is finite on the set 
\[
\mathscr Q = \{\rho\in \mathscr P\mid \rho\log\rho \in L^{1}\} ; 
\]
%\[
%\mathscr Q = \{\rho \in L\log L({\mathbb T_{L}^{d}}) : 
%\textstyle\int_{\mathbb T_{L}^{d}} \rho =1\} ; 
%\]
and there is no ambiguity to set 
\(\mathcal F_{\theta}(\rho) = +\infty\) for \(\rho\in \mathscr P
\setminus \mathscr Q\). 

Then, it suffices to show that for any \(\rho_1, \rho_2\in 
\mathscr Q\) the function \(s\mapsto \mathcal F_{\theta}(\rho_s)\)
where \(\rho_s = \rho_2 s + \rho_1 (1-s)\) is convex. It is 
straightforward to verify that \(\mathcal F_{\theta}(\rho_s)\)
is twice differentiable in \(s\in (0,1)\). Then, we compute
\[
\Big(\frac{d}{ds}\Big)^2 \mathcal F (\rho_s) = 
\int_{\mathbb T_{L}^{d}} \frac{\eta^2}{\rho_s}\,\, dx 
+ \theta L^d \int_{\mathbb T_{L}^{d}\times 
\mathbb T_{L}^{d}} V(x-y)\,\eta(x)\eta(y)\, dx\,dy,
\]
where \(\eta = \rho_2-\rho_1\). By Jensen's inequality 
we have 
\[
\int_{\mathbb T_{L}^{d}} \frac{\eta^2}{\rho_s}\,\, dx 
= \int_{\mathbb T_{L}^{d}} \Big(\frac{\eta}{\rho_s}\Big)^2\, \rho_s\, dx
\ge \Big(\int_{\mathbb T_{L}^{d}} 
\Big|\frac{\eta}{\rho_s}\Big|\, \rho_s\, dx\Big)^2
= \Big(\int_{\mathbb T_{L}^{d}} 
|\eta|\, dx\Big)^2.
\]
On the other hand, since \(|V(x-y)|\le V_{\rm max}\) we have
\[
\Big|\int_{\mathbb T_{L}^{d}\times\mathbb T_{L}^{d}} 
V(x-y)\,\eta(x)\eta(y)\, dx\,dy\Big|
\le 
V_{\rm max}\,\Big(\int_{\mathbb T_{L}^{d}} 
|\eta|\, dx\Big)^2. 
\] 
This implies the inequality
\[
\Big(\frac{d}{ds}\Big)^2 \mathcal F (\rho_s) 
\ge (1 - \theta L^d V_{\rm max})\,\Big(\int_{\mathbb T_{L}^{d}} 
|\eta|\, dx\Big)^2 >0
\]
if \(\theta L^d V_{\rm max} < 1\). 
\end{proof}
This immediately implies:
\begin{corollary} Under the conditions of 
Proposition~\ref{prop:conv}, \(\rho_0 = 1/L^d\) is 
the unique minimizer of \(\mathcal F_\theta (\rho)\). 
\end{corollary}

From the proof of the above Proposition we also obtain a corollary 
for potentials $V$ which are of {\it positive type}:

\bigskip
\noindent {\bf Definition} \hspace {2 pt}
A potential $V$ is said to be of {\it positive} type if for any 
 function $h\in L^{1}$, 
$$
\int_{\mathbb T_{L}^{d}\times \mathbb T_{L}^{d}}
\hspace{-6 pt}
V(x-y)h(x)h(y)dxdy \geq 0, 
$$
which is equivalent to the condition that $\forall k$
$$
\hat{V}(k) \geq 0.
$$
We let $\mathscr V^{+} \subset \mathscr V$ denote the set of interactions that are of positive type and, for future reference, the complimentary set by $\mathscr V_{N}$:
\begin{equation}
%\label{}
\mathscr V_{N}  =  \mathscr V\setminus \mathscr V^{+}.
\end{equation}

\begin{corollary}
Let $V \in \mathscr V^{+}$.  Then for all $\theta$, the unique minimizer of $\mathcal F_{\theta}(\cdot)$ is the uniform density  $\rho_{0}$.
\end{corollary}  
\begin{proof}
For $\rho  =  \rho_{0}(1 + \eta)$ in $\mathscr P$, we consider 
$f_{\eta}(s)
:= \mathcal F_{\theta}(\rho_{0}(1 + s\eta))$.  Calculating $f_{\eta}^{\prime\prime}(s)$ as in the proof of Proposition \ref{prop:conv}, the entropy term is still positive while the energy term yields
$$
\rho_{0}\theta\int_{\mathbb T_{L}^{d}\times \mathbb T_{L}^{d}}
\hspace{-6 pt}
V(x-y)\eta(x)\eta(y)dxdy \geq 0.
$$
Thus $f_{\eta}(s)$ is always convex and, for all $\eta$ always minimized at $s = 0$.  

Note that since all convexities are {\it strict},
any $\rho \in \mathscr P$ that is not a.e. equal to $\rho_{0}$ 
admits, for all $\theta$,
$$
\mathcal F_{\theta}(\rho_{0})
<
\mathcal F_{\theta}(\rho).
$$
\end{proof}

%\noindent\textcolor{magenta}{Here we should make some comment about exponential convergence to the minimizer.  \cite{JT2} 
%\cite{V} \cite{BGM}  Tamura proves if 2nd variation is uniformly positive then exponential convergence 
%-- something like that.  However, he may have had a confining potential to work with.
%Villani Carrillo and McCann prove exponetial convergence under particular convexity assumptions about the potential.
%}

Next we show that $V \in \mathscr V_{N}$ is also sufficient for the existence of a non--trivial phase.  The starting point is an elementary result which, strictly speaking is a corollary to Proposition \ref{KMO}.

\begin{proposition}
\label{KIU}
Let $V \in \mathscr V_{N}$ and suppose that at some $\theta_{d} < \infty$ there is a $\rho_{\theta_{d}}$ which is not a.e. equal to 
$\rho_{0}$ such that
$$
\mathcal F_{\theta_{d}}(\rho_{\theta_{d}})
\leq
\mathcal F_{\theta_{d}}(\rho_{0}).
$$ 
Then then for all $\theta > \theta_{d}$, $\rho_{0}$ is not the minimizer
of $\mathcal F_{\theta}(\cdot)$.
\end{proposition}
\begin{proof}
Indeed, since $\rho_{\theta_{d}}$ is not a constant it must be the case that
\begin{equation}
 \label{CXZ1}
\mathcal S(\rho_{\theta_{d}}) > \mathcal S(\rho_{0})
\end{equation}
thence
\begin{equation}
\label{CXZ2}
\mathcal E(\rho_{\theta_{d}}, \rho_{\theta_{d}}) < 
\mathcal E(\rho_{0}, \rho_{0}).
\end{equation}
Thus for $\theta > \theta_{d}$, it is seen that 
\begin{align*}
% \label{}
F_{\theta} &\leq \mathcal F_{\theta_{d}}(\rho_{\theta_{d}}) +
\frac{1}{2}\rho_{0}^{-1}(\theta - \theta_{d})
\mathcal E(\rho_{\theta_{d}}, \rho_{\theta_{d}})
\\
&< \mathcal F_{\theta_{d}}(\rho_{\theta_{d}})  +
\frac{1}{2}\rho_{0}^{-1}(\theta - \theta_{d})
\mathcal E(\rho_{0}, \rho_{0}) \leq \mathcal F_{\theta}(\rho_{0})
\end{align*}
which is the stated result.
\end{proof}

Thus beginning at $\theta = 0$ there is a non--trivial region or {\it phase} characterized by the property that 
$\rho_{0}$ is the unique minimizer for $\mathcal F_{\theta}(\cdot)$
and this phase terminates at some value of $\theta$ -- which is possibly infinite.  
Assuming this value is finite, we may refer to it as the {\it lower transition point} and above this point, there are non--trivial
minimizers of $\mathcal F_{\theta}$ and non--trivial solutions to 
Eqs.(\ref{VEV}) and (\ref{MV}).
(We shall refrain from naming this point till the possible nature of the transition at this point has been clarified.)  It should by noted, by a variant of the above argument, that {\it at} the lower transition point
$\rho_{0}$ is actually still a minimizer of the functional in Eq.(\ref{FE}).

We introduce some notation:

\noindent {\bf Definition} \hspace {2 pt}
For $V \in \mathscr V_{N}$, let $k^{\sharp}$ denote a minimizing wave vector for $\hat{V}(k)$:
$$
\hat{V}(k^{\sharp}) \leq \hat{V}(k) \hspace{15 pt}  \forall k.
$$
Note that $\hat{V}(k^{\sharp}) < 0$ by assumption.
We define $\theta^{\sharp} = \theta^{\sharp}(V)$ via
$$
\theta^{\sharp} := |\hat{V}(k^{\sharp})|^{-1}.
$$

We are finally ready for the following:
\begin{proposition}  {\rm{\cite{GP}}; see also \cite{BL}}  \hspace{2 pt}
\label{DmD}
Let $V \in \mathscr V_{N}$.
If $\theta > \theta^{\sharp}$ then $\exists \rho \in \mathscr P$ , $\rho \neq \rho_{0}$ which minimizes $\mathcal F_{\theta}(\cdot)$.  In particular, for $\theta > \theta^{\sharp}$, $\rho_{0}$ is no longer a minimizer of $\mathcal F_{\theta}$.  Thus 
$V \in \mathscr V_{N}$ is the necessary and sufficient condition for the existence of a non--trivial phase.
\end{proposition}

\begin{proof}
For $\theta > \theta^{\sharp}$ we may use as a trial minimizing function 
$$
\rho  =  \rho_{0}(1 + \varepsilon \eta^{\sharp})
$$
where $\eta^{\sharp}$ is a plane wave at wave number $k^{\sharp}$ and is itself of order unity
while $\varepsilon$ is to be regarded as a small parameter.
Since all quantities are bounded, we may expand:
\begin{align*}
% \label{}
\rho_{0}(1  + \varepsilon\eta^{\sharp})&
\log \rho_{0}(1  + \varepsilon\eta^{\sharp})  =  \\
&\rho_{0}(1  + \varepsilon\eta^{\sharp})\log \rho_{0} + \rho_{0}(1 + \varepsilon\eta^{\sharp})
(
\varepsilon\eta^{\sharp} - \frac{1}{2}[\varepsilon\eta^{\sharp}]^2) + o(\varepsilon^2).
\end{align*}
Since $\eta^{\sharp}$
integrates to zero,
\begin{equation}
\mathcal S(\rho) = \mathcal S(\rho_0) + 
\frac{1}{2}\varepsilon^2\rho_0\int|\eta^{\sharp}|^2dx
+o(\varepsilon^2).
\end{equation}
Meanwhile
\begin{align*}
\frac{1}{2}\theta\rho_{0}^{-1}\mathcal E(\rho,\rho) &=
\frac{1}{2}\theta\rho_{0}^{-1}\mathcal E(\rho_0,\rho_0)
+ \frac{1}{2}\varepsilon^2\theta\rho_{0}\int V(x-y)\eta^{\sharp}(x)
\eta^{\sharp}(y)dxdx
\\
&=  \frac{1}{2}\theta\rho_{0}^{-1}\mathcal E(\rho_0,\rho_0)
+ \frac{1}{2}\varepsilon^2\rho_{0}\hat{V}(k^{\sharp})||\eta^{\sharp}||_2^2
[\theta^{\sharp} + (\theta - \theta^{\sharp})].
\end{align*}
By definition of $\theta^\sharp$, 
$$
-\frac{1}{2}\varepsilon^2||\eta^{\sharp}||_2^2
=
\frac{1}{2}\theta^{\sharp}\varepsilon^2\hat{V}(k^{\sharp})||\eta^{\sharp}||_2^2
$$
so that 
$$
\mathcal F_{\theta}(\rho) = \mathcal F_{\theta}(\rho_{0})  -  
\frac{1}{2}\varepsilon^{2}\rho_{0}||\eta^{\sharp}||_{2}^{2}
\hspace{2 pt}
[|\hat{V}(k^{\sharp})|]
\hspace{2 pt}
(\theta - \theta^{\sharp}) + o(\varepsilon^{2})
$$
which is strictly less than $\mathcal F_{\theta}(\rho)$ for $\varepsilon$ sufficiently small.
(Here it is noted that the quantity $\rho_{0}||\eta^{\sharp}||_{2}^{2}$ is itself of order unity.)
By Proposition \ref{KIU}
the above is sufficient to establish the statement of this proposition.
\end{proof}

\begin{corollary}
For $V \in \mathscr V_{N}$, $\theta^{\sharp}(V)$ is the supremum of the set of quadratically stable parameter values for $\mathcal F_{\theta}(\rho_{0})$.  Furthermore, $\theta^{\sharp}$ marks the boundary for the linear stability of Eq.(\ref{MV}) with solution $\rho_{0}$. 
I.e.~for $\theta < \theta^{\sharp}$, $\rho_{0}$ is linearly stable while for 
$\theta > \theta^{\sharp}$, it is not.  
\end{corollary}
\begin{proof}
The first statement follows, in essence, from the above display.  As for the dynamics, the linearized version of Eq.(\ref{MV}) reads,
for $\eta  =  (\rho - \rho_{0})\rho_{0}^{-1}$
\begin{equation}
%\label{}
\frac{\partial \eta}{\partial t}  =  
\nabla^{2}(\eta + \theta V\star \eta).
\end{equation} 
The linear operator $\nabla^{2}[1 + \theta V\star](\cdot)$ has, by the definition of $\theta^{\sharp}$, a strictly negative spectrum if and only if 
$\theta < \theta^{\sharp}$.  The second statement of this corollary
therefore follows from the definition of linear order stability. 
\end{proof}

While the above is ostensibly vacuous, in these cases, it turns out that $\rho_{0}$ actually has a non--trivial {\it basin of stability} for $\theta < \theta^{\sharp}$.
\begin{theorem}
\label{ZZP}
Under the regularity assumption
$$
G  =  
\sum_{k}|\hat{V}(k)||k| < \infty
$$
there is a non--trivial basin of attraction for $\rho_{0}$ which contains all Borel measures that are sufficiently close to $\rho_{0}$ in the total variation distance.   In particular, at positive times, any such perturbing measure regularizes and, for any particular Sobolev norm, the density 
converges to $\rho_{0}$ exponentially 
fast in this norm.  The stated results hold uniformly in $L$.  

\end{theorem}

\noindent {\bf Remark} \hspace {2 pt}  The regularity assumption on $V$ is for convenience; presumably a stronger result is available.  In particular, it is not hard to see that with greater regularity of the perturbing density, regularity assumptions on the interaction potential can be relaxed. Moreover, with {\it greater} regularity assumptions on $V$, even more singular objects than Borel measures are contained in the basin of attraction.

\begin{proof}  
We write $\rho = \rho_{0}(1 + \eta)$ with $\eta$ measure valued.  Let $\hat{\eta}_{k}(t)$
denote the dynamically evolving $k^\text{th}$ Fourier mode. 
We shall assume that in the initial state, each mode is small -- which is certainly implied by the smallness of the total variation distance.
In particular, we will assume that at $t = 0$ each $\hat{\eta}_{k}(0)$ is bounded 
by an $\varepsilon_{0}$ which satisfies the condition that for all $k$,
\begin{equation}
\label{BVH}
2|k|\theta G\varepsilon_{0} < \lambda(k)
\end{equation}
where $\lambda(k)  =  k^{2}(1 - \theta \hat{V(k)})$ is the decay rate for the
the $k^{\text{th}}$ mode in the linear approximation.   It is emphasized that
$\lambda(k) > ck^{2}$ with $c > 0$ if 
$\theta < \theta^{\sharp}$.

The $\hat{\eta}_{k}(t)$ satisfy the formal equation
\begin{equation}
\label{EV}
\frac{\partial\hat{\eta}_{k}(t)}{\partial t}  =  -\lambda(k)
\hat{\eta}_{k}(t)  +  
\theta k\cdot \sum_{k^{\prime}}k^{\prime}\hat V(k^{\prime})
\hat{\eta}_{k^{\prime}}(t)
\hat{\eta}_{k^{\prime} - k}(t)
\end{equation}
where, it is reemphasized, all factors of volume have canceled out. It is noted that Eq.(\ref{EV}) may certainly be used to formulate dynamics via an iterative scheme -- provided that control is maintained under reiteration. Thus we 
may consider the sequence
$(\hat{\eta}^{(\ell)}_{k}(t)\mid \ell = 0, 1, 2, \dots)$
where $\hat{\eta}^{(0)}_{k}(t)
\equiv \hat{\eta}_{k}(0)$
and 
$\hat{\eta}^{(\ell + 1)}_{k}(t)$ is defined as the solution of 
Eq.(\ref{EV}) with $\hat{\eta}^{(\ell)}_{k}(t)$ the argument of the non--linear kernel.  

The form in which we will use equation \eqref{EV} is 
with moduli; we have
\[
\frac{\partial|\hat{\eta}_{k}(t)|}{\partial t}  =  -\lambda(k)
|\hat{\eta}_{k}(t)|  +  
\frac{1}{2}
\theta \Big[ \,\frac{\hat{\eta}_{k}(t)}{|\hat{\eta}_{k}(t)|}\,k\cdot \sum_{k^{\prime}}k^{\prime}\hat V(k^{\prime})
\hat{\eta}_{k^{\prime}}(t)
\hat{\eta}_{k^{\prime} - k}(t)
+ \text{c.c.}
\Big]
\]
The first claim is that for $\varepsilon_{0}$ satisfying the condition in 
Eq.(\ref{BVH}) then for all $k$ and $t$ and $\ell$,
$$
|\eta_{k}^{(\ell)}(t)|  \leq \varepsilon_{0}.
$$
Indeed, this is certainly true for $\eta_{k}^{(0)}$ so, inductively,
\begin{equation}
\label{RR0}
\frac{\partial | \eta_{k}^{(\ell)}(t)|}{\partial t} \leq -\lambda(k)|\eta_{k}^{(\ell)}(t)|
+ \theta\varepsilon_{0}^{2}G|k|.
\end{equation}
First, let us consider modes that satisfy
\begin{equation}
\label{RR}
|\eta_{k}(0)| > \frac{|k|\theta G \varepsilon_{0}^{2}}{\lambda(k)}.
\end{equation}
Such modes will decrease in magnitude -- at least till $|\eta_{k}(t)|$ reaches the 
right side of the inequality in Eq.(\ref{RR}) whereupon they may ``stick''.  But by assumption, these modes started out 
smaller than $\varepsilon_{0}$.  On the other hand, modes with initial conditions that satisfy the opposite inequality of Eq.(\ref{RR}) may actually grow till the inequality saturates but this does not get them past $\varepsilon_{0}$ since for all $k$,
\begin{equation}
\label{RR1}
\varepsilon_{0} > \frac{G\theta\varepsilon_{0}^{2}|k|}{\lambda(k)}
\end{equation} 
by hypothesis.  (The factor of two does not yet come into play.)  

Contraction of the sequence follows an identical argument which {\it does} employ the factor of 2.  We define
$\Delta_{k}^{\ell}(t) = |\hat{\eta}^{(\ell + 1)} - \hat{\eta}^{(\ell)}|$ and 
$\Delta_{\star}^{\ell}  =  \sup_{k,t}\Delta_{k}^{\ell}(t)$.  It is found that 
$\forall k,t$,
\begin{equation}
%\label{}
\Delta_{k}^{\ell}(t) \leq 
\frac{2\Delta_{\star}^{\ell-1}\varepsilon_{0}|k|\theta G}{\lambda(k)}
<(1-\delta)\Delta_{\star}^{\ell-1}
\end{equation}
for some $\delta > 0$ by Eq.(\ref{BVH}).
Thus Eq.(\ref{EV}) indeed defines our dynamics and we may perform manipulations on its basis
without further discussion.  Our next task will be to get the $\eta_{k}$ uniformly decaying.

By repeating the steps of Eqs.(\ref{RR0}) -- (\ref{RR1}) it is clear that
for any $e_{0} > \varepsilon_{0}$, there is a time $t_{0}$ such that for all $t > t_{0}$, 
\begin{equation}
\label{VGV}
|\eta_{k}(t)|  <  \frac{G\theta e_{0}^{2}|k|}{\lambda(k)}.
\end{equation}
Incidentally, we have now placed $\eta$ in some reasonable 
Sobolev space -- but this is not yet relevant.
For the moment, the pertinent observation 
is that there is a maximum sized mode which is to be found 
at a finite value of $k$ (which may, of course, change from time to time).
Thus, for each $t > t_{0}$, let $\beta_{0}$ denote 
the modulus of the maximum mode and $\overline{k}$ denote the wave vector that maximizes.  Then, for all $t > t_{0}$ we have
\begin{equation}
\label{TDR}
\frac{\partial\beta_{0}}{\partial t}
\leq
-\lambda(\overline{k}) \beta_{0} + G\theta|\overline{k}|\beta_{0}\varepsilon_{0} \leq -\frac{1}{2}\lambda_{\text{min}}\beta_{0}
\end{equation}
where $\lambda_{\text{min}}$ is the minimum of $\lambda(k)$
(which is positive for $\theta < \theta^{\sharp}$).
We conclude that all the $\eta_{k}(t)$ tends to zero exponentially fast with rate at least as large as 
$\frac{1}{2}\lambda_{\text{min}}$.  

We use a small variant of this argument to show that for any
$n$, the maximum of $|k|^{n}|\eta_{k}(t)|$ (exists and) decays with a rate at least as large as $\frac{1}{2}\lambda_{\text{min}}$.  
Focusing on $n \geq 1$ let us assume that at the $n-1^{\text{st}}$
stage of the argument, we have a $t_{n-1}$ such that for all 
$t > t_{n-1}$, 
\begin{equation}
%\label{}
\beta^{[n-1]}_{k}(t)  \leq  \frac{\theta|k|G}{\lambda(k)}
\delta_{n-1}2^{n-1}
\end{equation}
here $\beta^{[n]}_{k}(t) := |k|^{n}|\eta_{k}(t)|$ and the quantity
$\delta_{n}$ is specified as follows:  
Multiplying both sides by $|k|$, since $\lambda \geq ck^{2}$ this puts a uniform bound on $\beta^{[n]}_{k}(t)$ which is stipulated to be less than one.  

We now write
\begin{align*}
%\label{}
\frac{\partial \beta_{k}^{[n]}}{\partial t}
\leq -\lambda(k)\beta_{k}^{[n]} +
\hspace{3 pt}&2^{n-1}|k|\theta\sum_{k^{\prime}}\beta_{k^{\prime}}^{[n]}
|\eta_{k^{\prime} - k}||\hat{V}(k^{\prime})k^{\prime}|
 + \\
&2^{n-1}|k|\theta\sum_{k^{\prime}}\beta_{k^{\prime} - k}^{[n]}
|\eta_{k^{\prime}}||\hat{V}(k^{\prime})k^{\prime}|
\end{align*}
where we have used 
$|k|^{n}  =  |k^{\prime} + k - k^{\prime}|^{n}
\leq  2^{n-1}[|k|^{n} + |k - k^{\prime}|^{n}].$
We now wait till a time $t_{n}^{\prime}$ when each $|\eta_{t}(k)|$
is less then some $\varepsilon_{n}$ which is small and to be specified.  Summing, the estimates,
\begin{equation}
%\label{}
\frac{\partial \beta_{k}^{[n]}}{\partial t}
\leq -\lambda(k)\beta_{k}^{[n]} + 2^{n}\theta|k|\varepsilon_{n}
\end{equation}
and we now wait till a time $t_{n} > t_{n}^{\prime}$ so that, similar to the previous portion of the argument,
\begin{equation}
\label{StA}
\beta_{k}^{[n]}(t)  \leq  \frac{\theta |k| G}{\lambda(k)} 
2^{n}\delta_{n}
\end{equation} 
for any $\delta_{n} > \varepsilon_{n}$.  Obviously this 
$\delta_{n}$ will be tailored to satisfy the requirements to propagate the 
{\it next} iterate of the argument -- and so we stipulate.  
But in addition we require that $2^{n}\delta_{n} < \varepsilon_{0}$

The proof is completed by noting that Eq.(\ref{StA}) allows us to conclude that the supremum of $\beta_{k}^{[n]}$ is to be found at a finite $k$
and the rest of the argument proceeds as described in the vicinity of 
Eq.(\ref{TDR}).  The desired result has been proved.
\end{proof}
%Starting at some time $t_{2}$, we may now assume that 
%$|\hat{\eta}_{k}(t)| < e_{0}$ for all $t > t_{2}$ with 
%$e_{0}$ sufficiently small to be determined momentarily.  
%Let $\alpha_{k}(t) := |k|^{a} \hat{\eta}_{k}(t)$ where 
%$a \in (\frac{1}{2}, 1)$.
%We define, as before, $\overline{\alpha}(t)$ to be the maximum of 
%$|\alpha_{k}(t)|$ which, in light of Eq.(\ref{VGV})
%is well defined and occurs at some finite $\underline{k}$.
%Starting back at Eq.(\ref{EV}) we arrive, for $t > t_{2}$, at 
%\begin{equation}
%\label{SDQ}
%\frac{\partial \overline \alpha}{\partial t}
%\leq
%-\lambda(\underline k)\overline{\alpha} + |\underline{k}|^{1 + a}e_{0}
%H_{a}\overline{\alpha}
%\end{equation}
%where $H_{a}$ is the sum of $\hat{V}(k)|k|^{1-a}$.  
%Since $\lambda(k) > ck^{2}$ it is seen that if 
%$e_{0}$ is sufficiently small, the coefficient 
%of $\overline{\alpha}$ on the right 
%of Eq.(\ref{SDQ}) is negative 
%regardless of (uniformly in) the particular value of the maximizing
%$\underline k$.  From this the exponential convergence 
%of $\eta(x,t)$ to zero in  $L^{2}$ follows readily.
%\end{proof}

%\noindent\textcolor{magenta}{Actually, this just works in 
%1D.  Must recycle several more times but then get
%$H^{n}$ convergence.}

\subsection{Phase transitions in the V--McK systems (2): \\ First  order transitions.}

In order to investigate the possibility of continuous/discontinuous transitions in this model, an appropriate definition must be provided.

\noindent {\bf Definition} \hspace {2 pt}
Consider the V--McK functional $\mathcal F_{\theta}$ with $V\in \mathscr V_{N}$.  We define 
$\theta_{c}$ to be a (lower) critical point if the following criteria are satisfied:

$\bullet$  For $\theta \leq \theta_{c}$, $\rho_{0}$ is the unique minimizer of $\mathcal F_{\theta}(\cdot)$.

$\bullet$  For $\theta > \theta_{c}$, $\exists \rho_{\theta} \neq \rho_{0}$ which minimizes $\mathcal F_{\theta}(\cdot)$

\noindent  (which necessarily implies that $\rho_{0}$ is no longer a minimizer of $\mathcal F_{\theta}(\cdot)$).

$\bullet$  If $(\rho_{\theta}\mid \theta > \theta_{c})$ is any family of such minimizers then 
$$
\limsup_{\theta \downarrow \theta_{c}}
||\rho_{0} - \rho_{\theta} ||_{1}  = 0.
$$

\noindent {\bf Remark} \hspace {2 pt}
We have called this a lower critical transition since, conceivably there could be later (in $\theta$) transitions of this type with non--trivial solutions ``bifurcating'' from preexisting non--trivial solutions.  This would be difficult to detect -- analytically or numerically -- since the non--trivial solutions are anyway evolving with $\theta$.  Such a phenomenon would, presumably, have to be tied to non--analyticity in $\mathcal E(\cdot)$ or 
$\mathcal S(\cdot)$ notwithstanding their {\it continuity}.  By contrast 
(c.f. Proposition \ref{XDQ} below) for the other possible type of transition, these objects are generically {\it discontinuous}.  In any case, the foremost possible phase transition in these systems is the lower one and will be the focus of all our attention.  

\vspace{6 pt}

Any (lower) phase transition not satisfying the above three items will be called a {\it discontinuous transition} and we will denote will denote such a transition point by $\theta\t$.  
As we shall see later, in Proposition \ref{XDQ}, for a discontinuous transition the second item will hold
while in the first item, we must replace $\theta \leq \theta_c$ with
$\theta < \theta\t$.  But most pertinently, the third item fails in its entirety.  
Thus, at such a transition point, a new minimizing solution of Eq.(\ref{MV}) appears which,
for $\theta = \theta\t$, is degenerate (in the sense of minimizing 
$\mathcal F_{\theta}(\cdot)$) with $\rho_{0}$ but is markedly separated from $\rho_{0}$.  

%\noindent\textcolor{green}
%{ 
%\begin{theorem}  For the K--McV systems on $\mathbb T_{L}^{d}$ there is never a critical ($\theta_{c}$) transition.  In particular for $V \in \mathscr V_{N}$ there is always a discontinuous transition at some finite value $\theta\t$.  At $\theta\t$ there are non--trivial minimizers one of which will be denoted by $\rho_e$.  Then 
%$\mathcal F_{\theta_{T}}(\rho_{0})  = \mathcal F_{\theta_{T}}(\rho_{e})$
%while $\mathcal S(\rho_{0}) < \mathcal S(\rho_{e})$ and 
%$\mathcal E(\rho_{0}, \rho_{0}) > \mathcal E(\rho_{e}, \rho_{e})$.  I.e. 
%\noindent\textcolor{magenta}{make precise notion of $E_{\theta}$, $S_{\theta}$ and show these are discontinuous.}  Moreover, $\theta\t < \theta^{\sharp}$.
%\end{theorem}
%}

%\noindent\textcolor{purple}{This to be rewritten as some sort of a proposition
%which states the first half of the proof.  Then, ``main theorem'' will concern
%$d \geq 2$.}

\vspace{4 pt}

Our first result characterizes the critical transitions:

\begin{proposition}
\label{ZCC}
Let $V \in \mathscr V_{N}\cap L^{2}$ and
suppose a (lower) critical phase transition as described above occurs in the V-McK system at some $\theta_{c}$.  Then, necessarily, $\theta_{c}  =  \theta^{\sharp}$.
\end{proposition}

\begin{proof}  
The trivial cases $\theta^{\sharp} = 0$ or $\theta^{\sharp} = \infty$ are easily dispensed with.  Assuming otherwise for $\theta^{\sharp}$, it is 
obvious that a (lower) critical point $\theta_{c}$ could not exceed $\theta^{\sharp}$ since non--trivial minimizers already exist at any $\theta > \theta^{\sharp}$.  We shall therefore work with $\theta < \theta^{\sharp}$ and write $\theta = \theta^{\sharp} - \delta$ where $\delta  > 0$.

As a preliminary, it should be noted that while the third item in the definition of the $\theta_{c}$ necessarily reflects the natural $L^{1}$--norm, it will be more convenient to work with $L^{2}$ and $L^{\infty}$.  
We will show that, as far as $\rho_{\theta} - \rho_{0}$
is concerned, these are controlled by the $L^{1}$--norm.
First, for expositional ease, let us define
\begin{equation}
%\label{}
\eta_{\theta} := \frac{\rho_{\theta} - \rho_{0}}{\rho_{0}}.
\end{equation}
Starting with $L^{2}$, recall from the ``obvious corollary'' to  \ref{BbbB} that 
since $\rho_{\theta}$ is a minimizer of 
$\mathcal F_{\theta}(\cdot)$ it is bounded uniformly (enough) in $\theta$ and thence 
$\eta_{\theta}$ is similarly bounded
 -- say by $\omega$.  
Then
\begin{equation}
\label{MTS}
||\eta_{\theta}||_{2}^{2}  \leq  
||\eta_{\theta}||_{1}||\eta_{\theta}||_{\infty}
\leq
\omega||\eta_{\theta}||_{1}
\end{equation}
For the moment, we can only employ the outer inequality but at least we now have that $||\eta_{\theta}||_{2}$
is ``small''.  Next we use the fact that $\rho_{\theta}$ satisfies the Kirkwood--Monroe equation, Eq.(\ref{VEV}).   As is not hard to see, in the language of $\eta_{\theta}$ this reads
\begin{equation}
%\label{}
1 + \eta_{\theta}(x)  =  \frac{\text{e}^{-[\theta V\star \eta_{\theta}](x)}}
{\int \text{e}^{-\theta V\star \eta_{\theta}}\rho_{0}dx}.
\end{equation}
Now, for a.e.~$x$
\begin{equation}
%\label{}
|[V\star\eta_{\theta}](x)|  =  \left |
\int_{\mathbb T_{L}^{d}}
\hspace{-5 pt}
V(x-y)\eta_{\theta}(y)dy \right |
\leq  ||V||_{2}||\eta_{\theta}||_{2}
\end{equation}
thence, if $\eta_{\theta}(x) > 0$,
\begin{equation}
\label{KCX1}
\eta_{\theta}(x)  \leq  (\text{e}^{2\theta||V||_{2}||\eta_{\theta}||_{2}} -1)
\end{equation}
while if $\eta_{\theta}(x) < 0$,
\begin{equation}
\label{KCX2}
\eta_{\theta}(x)  \geq  (\text{e}^{-2\theta||V||_{2}||\eta_{\theta}||_{2}} -1)
\end{equation}
Thus, for $||\eta_{\theta}||_{2}$ sufficiently small
(which we know happens as $\theta \downarrow \theta_{c}$ from Eq.(\ref{MTS}))
there is a $K$ -- which is uniform in $\theta$ near $\theta_{c}$ and of order unity -- such that $||\eta_{\theta}||_{\infty}  <  K||\eta_{\theta}||_{2}$.  We may now exploit the middle inequality in Eq.(\ref{MTS}) and declare that in the vicinity of the purported $\theta_{c}$ all norms of any 
$\eta_{\theta}$ are comparably small.

Now, suppose that $\theta \gtrsim \theta_{c}$.  
We repeat the calculations performed in
Proposition \ref{DmD} with the result that
\begin{equation}
%\label{}
\mathcal F_{\theta}(\rho_{\theta})  =
\mathcal S(\rho_{0}) 
+ \frac{1}{2}\theta L^{d} \mathcal E(\rho_{0},\rho_{0}) +
\frac{1}{2}\rho_{0}
\left [
||\eta_{\theta}||^{2}_{2} + 
\theta\mathcal E(\eta_{\theta}, \eta_{\theta})
\right ] 
+ o(||\eta_{\theta}||^{2}_{2})
\end{equation}
The term in the square brackets is strictly positive and at least of the order
$||\eta_{\theta}||^{2}_{2}$
if $\theta \approx \theta_{c} = \theta^{\sharp} - \delta $ with $\delta  > 0$.  
Evidently, as indicated, the only possibility for a continuous transition is at $\theta^{\sharp}$. 
\end{proof}

The alternative to a critical transition is a {\it discontinuous} transition which is also called a first order transition.  For such transitions, the following holds: 
\begin{proposition}
\label{XDQ}
If $V \in \mathscr V_{N}$ and the criteria in the preceding definition of a critical point fails then there is a transition at some $\theta\t$ which is characterized by the following:

$\exists \rho_{\theta\t} \neq \rho_{0}$ such that

\noindent $\bullet \hspace{2 pt}$ $\mathcal F_{\theta\t}(\rho_{\theta\t})  
=  \mathcal F_{\theta\t}(\rho_{0})  =  F_{\theta\t}$

\noindent $\bullet \hspace{2 pt}$  $\mathcal E(\rho_{\theta\t},\rho_{\theta\t})  <  
\mathcal E(\rho_{0},\rho_{0})$

\noindent $\bullet \hspace{2 pt}$  $\mathcal S(\rho_{\theta\t})  
>  \mathcal S(\rho_{0})$

\noindent (and thus both $E_{\theta}$ and $S_{\theta}$ are discontinuous at 
$\theta = \theta\t$).
\end{proposition}

Since two distinctive minimizers exist at the same value of $\theta$, such a point may also be described as a point of {\it phase coexistence}.  

\begin{proof}
At $\theta > \theta\t$ we have for $\eta_{\theta} = (\rho_{\theta} - \rho_{0})\rho_{0}^{-1}$
\begin{equation}
%\label{}
\limsup_{\theta \downarrow \theta\t}||\eta_{\theta}||_{1} \neq 0.
\end{equation}
Since, in these matters, all norms are more or less equivalent, we will take the above statement in $L^{2}$ and extract a weakly convergent sequence which we will still denote by $\eta_{\theta}$.  Let us first rule out the possibility that 
$\eta_{\theta} \rightharpoonup 0$.  Indeed, supposing this to be the case, we would certainly have
$$
\lim_{\theta \to \theta\t}\mathcal E(\eta_{\theta},\eta_{\theta})  =  0.
$$  
e.g.,~as discussed in the proof of Theorem \ref{MMM}.
However, we have that all along the subsequence, $||\eta_{\theta}||_{2} \geq h\t$ for some
$h\t > 0$
and, moreover, for some $b < \infty$,  $||\eta_{\theta}||_{\infty} < b$.  
Thence, by the convexity properties of the $\mathcal S$--term we have that for $s$ small, 
$$
\mathcal S(\rho_{\theta})  \geq  \mathcal S(\rho_{0}) + 
\frac{1}{2}s^{2}||\eta_{\theta}||^{2}_{2}  + o(s^{2}).
$$
This indicates that 
$$
\limsup_{\theta\downarrow \theta\t} F_{\theta} > F_{\theta\t}
$$
in violation of the stated continuity result.

Thus, in our sequence $\eta_{\theta}$ converges to a non--trivial limit which we  
 denote (optimistically) by $\eta_{\theta\t}$.  On the energetic side, we still have
$$
\lim_{\theta \to \theta\t}\mathcal E(\eta_{\theta},\eta_{\theta})  =  
\mathcal E(\eta_{\theta\t},\eta_{\theta\t})
$$  
and, again, by convexity properties,  $\mathcal S(\rho_{0}(1 + \eta_{\theta\t}))$ does not exceed any limit of 
$\mathcal S(\rho_{\theta})$ as $\theta \downarrow \theta\t$.
Evidently this $\eta_{\theta\t}$ provides a genuine minimizer for 
$\mathcal F_{\theta\t}(\cdot)$ which we now denote by $\rho_{\theta\t}$.

By hypothesis (of a {\it lower} transition) the uniform solution
is a minimizer of $\mathcal F_{\theta}$ up to $\theta = \theta\t$ and thus by continuity is also a minimizer {\it at} $\theta\t$:
 $\mathcal F_{\theta\t}(\rho_{0})  =  \mathcal F_{\theta\t}(\rho_{\theta\t})$
(see Proposition \ref{KMO}).
 Moreover, we reiterate, $\mathcal S({\rho_{0}}) < \mathcal S(\rho_{\theta\t})$
necessarily implying
$\mathcal E(\rho_{\theta\t}, \rho_{\theta\t}) < \mathcal E(\rho_0, \rho_{0})$. 
All of the stated results have now been proven.

\end{proof}

The two preceding results -- concerning (i) the purported critical behavior 
at $\theta = \theta^{\sharp}$ 
and (ii) the characteristics of systems with purported non--critical lower transitions -- 
allow for the following:

\begin{theorem}
\label{BD}
Consider, in dimension $d \geq 2$ a fixed $V \in \mathscr V_{N}$ which is isotropic.
Then, if the volume is sufficiently large
there is never a (lower) critical transition.
In particular under the above stated conditions there is a 
discontinuous transition at some
$\theta\t$ satisfying 
$\theta\t < \theta^{\sharp}$ where there is phase coexistence and various other properties all of which has been described in the context of Proposition \ref{XDQ}.

\end{theorem}

\begin{proof}
We will consider disturbances of the form
$$
\rho  =  \rho_{0}(1  + {\varepsilon}\eta)
$$
with $\eta(x)$ a function (with $L^{\infty}$--norm) of order unity and $\varepsilon$ a small (pure) number of order unity.
Then
\begin{equation}
%\label{}
\mathcal S(\rho)  =  \mathcal S(\rho_{0}) +
\int_{\mathbb T_{L}^{d}}
\hspace{-1 pt}
\rho_{0}[ \frac{1}{2}\varepsilon^{2}\eta^{2} -\frac{1}{6}\varepsilon^{3}\eta^{3}]dx + o(\varepsilon^{3})
\end{equation}
where it is slightly important to observe that $o(\varepsilon^{3})$ is independent of the volume.
Of course the above expansion also contained a 
{\it linear} odd term which vanishes due to symmetry.  Similarly, we have $\mathcal E(\rho,\rho) - \mathcal E(\rho_{0},\rho_{0}) = \varepsilon^{2}\mathcal E(\rho_{0}\eta,\rho_{0}\eta)$.

We set $\theta = \theta^{\sharp}$ where, as we recall, the minimizing wave vector satisfies
$-\hat{V}(k^{\sharp})\theta^{\sharp}  =  1$.  Now, we invoke the assumption that $V(x)$ depends only on 
$|x|$, -- so that $\hat{V}(k)$ depends only on $|k|$.   Then, under the
auspices of 
continuous wave numbers (``the infinite volume limit'') we could find $\tilde{k}_{1}$ and $\tilde{k}_{2}$
with $|k^{\sharp}|  =  |\tilde{k}_{1}|  =  |\tilde{k}_{2}|$ necessarily satisfying
\begin{equation}
\label{ZDF1}
\hat{V}(k^{\sharp})  =  \hat{V}(\tilde{k}_{1})  =  \hat{V}(\tilde{k}_{2})
\end{equation} 
such that
\begin{equation}
\label{ZDF2}
k^{\sharp} + \tilde{k}_{1} + \tilde{k}_{2} = 0.
\end{equation}
Thus, in finite volume, we can find approximating $k_{1} \approx \tilde{k}_{1}$ and $k_{2} \approx \tilde{k}_{2}$
with, e.g., $|k_{1} - \tilde{k}_{1}| = O(L^{-1})$ that are appropriate to $\mathbb T_{L}^{d}$
such that Eq.(\ref{ZDF2}) is true and Eq.(\ref{ZDF1}) is approximately true.  We now use 
$$
\eta  =  \eta^{\sharp} + \eta_{1} + \eta_{2}
$$
with $\eta_{1}$ and $\eta_{2}$ plane waves at wavenumbers $k_{1}$ and $k_{2}$ respectively. 
We have, e.g.,
\begin{equation}
%\label{}
\hat{V}(k_{1})\theta^{\sharp}\rho_{0}||\eta_{1}(k_{1})||_{2}^{2}  +  \rho_{0}||\eta_{1}||^{2}_{2}
\leq \sigma(L)
\end{equation}
with $\sigma \to 0$ as $L \to \infty$.
(We reiterate that each term in the above display is separately of order unity.)
Thus, we may declare that, essentially, up through second order
$\mathcal F_{\theta^{\sharp}}(\rho_{0}(1 +{\varepsilon}\eta))$ equals $\mathcal F_{\theta^{\sharp}}(\rho_{0})$.
But now, since $k^{\sharp} + \tilde{k}_{1} + \tilde{k}_{2} = 0$, then unlike a plane wave which, even cubed, would integrate to zero, it is in general the case that
\begin{equation}
%\label{}
\int_{\mathbb T_{L}^{d}}(\eta^{3})dx \neq 0.
\end{equation}
Moreover by adjusting the phases of the constituents, 
the corresponding term in the expansion of 
$\mathcal F_{\theta^{\sharp}}(\rho)$
 can be made to be negative.  

It is thus evident that for all $L$ large enough, $\rho_{0}$ is not the minimizer
for $\mathcal F_{\theta^{\sharp}}(\cdot)$ and by the continuity 
result contained in Proposition \ref{KIU},
$\rho_{0}$
is not the minimizer for a range of $\theta$ which lies strictly below $\theta^{\sharp}$.  
Thus the transition takes place at some $\theta\t < \theta^{\sharp}$ and is 
(therefore) not continuous.

\end{proof}

We conclude this section with an (abbreviated) spectrum of remarks.

\vspace{1 mm}

\noindent{\bf Remark} \hspace {2 pt}
For the vast majority of physically motivated single component systems, the above theorem precludes, in the general context, the possibility of continuous transitions.  (Cf.~the third remark in this sequence for additional discussion.)   This is in apparent contradiction with a number of results for these system -- some of which receive additional discussion in the subsequent remark -- 
the most pertinent of which have been the subject of \cite{CCELM} 
and, recently, discussed in \cite{BL} .  
In these works, a continuous transition was indeed found at the analog of $\theta^{\sharp}$.  The important distinction  
distinction between the present work and \cite{CCELM}, \cite{BL} is in the nature of the entropy functional that was employed.  Indeed, therein the prototypical entropy functional was of the form
\begin{equation}
%\label{}
\mathcal A\hspace{.5 pt}_{0}(\rho)  =  
\int [\rho\log\rho  +  (1-\rho)\log(1-\rho)]dx.
\end{equation}
Thus, in the expansion which uses $\rho = \rho_{0}(1 + \eta)$, all the odd terms in $\eta$ vanish identically;  from this perspective, $\mathcal A\hspace{.5 pt}_{0}(\rho)$
is simply the symmetrized version of $\mathcal S(\rho)$.
Of course this preempts the term(s) driving the conclusion of Theorem \ref{BD} and thus allows for a continuous transition at $\theta = \theta^{\sharp}$.  

However, $\mathcal A\hspace{.5 pt}_{0}(\rho)$ is not a natural entropy form for a one--component system and, as argued in \cite{BL}, is in fact an {\it effective}
entropy term for a two component Ising--type system.  The first principles version of these sorts of Ising systems is currently a subject of intensive investigation; e.g.~the works \cite{CCELM} and \cite{EGM} and some work in progress by 
R. Esposito and R. Marra in conjunction with the authors.  The present set of models under consideration appear to undergo a transition that is, at least sometimes, ``weakly first--order'' at some 
$\theta\t \lesssim \theta^{\sharp}$.  
However, it may well be the case that the consideration of more general interactions leads, in the two--component cases, to generic circumstances where there are continuous transitions.  

\vspace{2 mm}

\noindent{\bf Remark} \hspace {2 pt}  In a variety of contexts, e.g.~\cite{GK}, \cite{JT} \cite{MA}, various workers have claimed that non--trivial solutions to Eq.(\ref{VEV}) grow -- continuously or otherwise -- only after $\theta \geq \theta^{\sharp}$.  Since this attitude seems prevalent, perhaps some comments are in order.  As is typical -- by definition -- for discontinuous transitions the new minimizing solutions or {\it states} are {\it not} continuously connected to the old.  This and, especially,
Theorem \ref{ZZP}
accounts for some of the difficulties attempting to generate stable solutions dynamically as described in 
\cite{MA}.  It may, perhaps, prove useful to attempt to generate the stable solutions by a nucleation technique 
(e.g., based on the solution for small systems)
at parameter values {\it below} $\theta^{\sharp}$ where, perhaps, fewer interfering solutions exist.  Indeed, the basin of attraction provided by Theorem
\ref{ZZP} may itself, for $L \gg 1$, be arguably small from a certain perspective.  

\vspace{-2 pt}

The results in  \cite{GK} and \cite{JT} both rely on standard fixed point/bifurcation analyses.  In
\cite{GK}, it was simply {\it assumed} that the non--trivial solutions were periodic with period 
$k^{\sharp}$.  Of course, as was noted in 
\cite{GK} the Kirkwood--Monroe equation is 
``closed''  under periodic functions with 
{\it any} period.  (By this it is meant that if we write
Eq.(\ref{VEV}) in the form 
$\rho  =  \Xi(\rho)$  then, if $\sigma$ is periodic so is 
$\Xi(\sigma)$ with the same period.) Thus, using this equation as the basis for a fixed point argument (with the help of the Krasnoselskii fixed point theorem) one {\it is} liable to manufacture a solution of sorts.  Moreover, this scheme indeed requires 
$\theta \geq \theta^{\sharp}$ for the solution of period $2\pi[k^{\sharp}]^{-1}$ to be non--trivial.  However it is also clear that these solutions have no stability under the dynamics of
Eq.(\ref{MV}) and, even, the discrete time dynamics which produced these solutions in the first place.  In particular it is almost certain that these solutions do not minimize the free energy functional.  (Although, no doubt, they have a lower free energy than the uniform solution.)  Indeed while it is not impossible that the stable solutions 
appearing at $\theta = \theta\t$ are periodic with 
{\it some} period, as is typical in non--linear phenomena, there is no reason for the period to exactly match that of the unstable mode which appears at $\theta = \theta^{\sharp}$.  

Finally, we wish to comment on the careful analysis in \cite{JT}.  Here, standard results in the theory of bifurcations were brought to bear under the explicitly stated assumption that the relevant hypotheses for the theorem actually apply.  In this context, the most important of these ingredients is that the kernel and co--kernel of the linear operator are one--dimensional.  Under the required symmetry 
$V(x-y)  =  V(y-x)$ -- without which the model   
does not make sense as a description of identical interacting constituents -- this condition is obviously violated.  And it may or may not be a ``technical'' violation, cf.~the next remark.  Notwithstanding, even if the conditions for the bifurcation results are satisfied, provided that 
$\hat{V}(k)$ is continuous, there are always modes 
near $k^{\sharp}$ which are nearly unstable at 
$\theta = \theta^{\sharp}$.  Thus any basin of stability and domain of validity will be vanishingly small with increasing volume.

%\noindent\textcolor{juice}{ Appears to be in contrast with other things that have been proved in the literature.
%(1)  Japanese theses result.  (2) Result of Grewe and Klein 
%%  (16) Grewe, N.; Klein, W. J. Math. Phys. 1977, 18, 1735. 
%which get bifurcation from 
%$\theta^{\sharp}$}.

\vspace{2 mm}

\noindent {\bf Remark} \hspace {2 pt}
It is remarked that full isotropy of $V(x)$ and/or
$d > 1$ is not strictly required.  The condition used in the preceding proof was the existence of three wave vectors adding up to zero each of which (nearly) minimize $\hat V(k)$.  Obviously this can be achieved in $d > 1$ if $V(x)$ has an appropriate 
3--fold symmetry.  Moreover, a detailed analysis will yield alternative sufficient conditions:  (I)  If 
$\hat V(0)$ (assumed positive) is not too large.  (II) if 
$\hat{V}(2k^{\sharp})$ is negative and, in magnitude, an appreciable fraction of $\hat{V}(k^{\sharp})$; etc., etc.  
 However, full isotropy is not an unreasonable assumption for fluid systems -- as well as other applications -- and, in fact, $d >1$ is required for actual statistical mechanics systems with short--range interactions to exhibit changes of state.  Thus we are content with the present result and will not pursue these alternative specialized circumstances.

\vspace{2 mm}

\noindent{\bf Remark} \hspace {2 pt} In the language of equilibrium statistical mechanics, $\theta\t$ is, of course, classified as a point of first order transition while the point $\theta^{\sharp}$ is not recognized.  
From the perspective of dynamical systems, $\theta^{\sharp}$ is a subcritical pitchfork bifurcation.
It may be presumed that solutions of the type which minimize at and above $\theta\t$ are present even before $\theta\t$.  The point at which they first appear -- temporarily denoted by
$\theta_{\text{R}}$ -- would then represent a {\it saddle node bifurcation}
while, from this perspective, the point $\theta\t$ is not recognized.

{\section {The large volume limit}

In this final section, we shall investigate the behavior of our systems -- with fixed $V(x)$ -- as $L$ tends to infinity.  The upshot, roughly speaking, is that for interaction potentials which are appropriate for physical problems the energy/temperature scaling is viable and not so otherwise.  
Since the $L$--dependence of these problems will now be our focus, all relevant quantities will adorned with superscript $[L]$.

\subsection{The limit of the transition points}

If the interaction violates the conditions of Theorem 
\ref{BD} and has a (sequence of) continuous transitions then, by Proposition \ref{ZCC}, these all take place at the relevant $\theta^{\sharp}$ which is only weakly dependent on system size.  Thus the interesting questions concern the discontinuous transitions.  
Notwithstanding, the forthcoming makes no explicit use of the discontinuity other than the convenience of label.

\begin{theorem}
For fixed $V \in \mathscr V_{N}$ consider the system on $\mathbb T_{L}^{d}$ with discontinuous transition at $\theta\t^{[L]}$.  Then these transition points tend to a definitive limit.  
%Moreover, the transition points are essentially decreasing in the sense that for fixed $L_{1}$, and any $\epsilon$,
%$$
%\theta\t^{[L_{1}]} > + \theta\t^{[L_{2}]} + \epsilon
%$$
%for all $L_{2}$ sufficiently large.
\end{theorem}
\begin{proof}
We shall start with the statement that for any $L$ and any integer $n$,
\begin{equation}
%\label{}
\theta\t^{[L]}  \geq  \theta\t^{[nL]}.
\end{equation}
To see this, we patch together $n^{d}$ copies of a non--trivial minimizer of
$\mathbb T_{L}^{d}$ at $\theta\t^{[L]}$ to cover $\mathbb T_{nL}^{d}$
(which is facilitated by the fact that, anyway, these solutions are periodic).    
First, letting 
\begin{equation}
%\label{}
v = \int_{\mathbb R^{d}}
\hspace{-5 pt}
V(x)dx
\end{equation}
it is noted that for any $L_{a}$,
\begin{equation}
\label{LA}
\mathcal F_{\theta}^{[L_{a}]}(\rho_{0}^{[L_{a}]})  =  -\log L_{a}^{d}  +  \frac{1}{2}\theta \hspace{1 pt}v.
\end{equation}
Now let $L$ denote any scale with transition temperature 
$\theta\t^{L}$ and let $\rho_{\star}^{[L]}$ denote the non--trivial minimizer
for $\mathbb T_{L}^{d}$ at this value of the parameter.  
Let $\tilde \rho_{\star}^{[nL]}$ denote the periodic extension of this function to 
$\mathbb T_{nL}^{d}$ rescaled by a factor of $n^{-d}$ so that it is properly normalized.  
It is seen that
\begin{align*}
%\label{}
\mathcal S^{nL}(\tilde \rho_{\star}^{[nL]})  
&=  \int_{\mathbb T_{nL}^{d}} \tilde \rho_{\star}^{[nL]} \log \tilde \rho_{\star}^{[nL]} dx \\
&= -\log n^{d}  +  n^{d}\times\frac{1}{n^{d}}
  \int_{\mathbb T_{L}^{d}} \rho_{\star}^{[L]} \log \rho_{\star}^{[L]} dx \\
& = -\log n^{d}  + \mathcal S^{[L]}(\rho_{\star}^{[L]}).
\end{align*}

Making use of the underlying periodic structure, we have that for fixed $L$--periodic $g(y)$,
the integral $\int_{\mathbb T_{nL}^{d}}V(x-y)g(y)dy$ is equal to the periodic extension of the corresponding integral on $\mathbb T_{L}^{d}$.  Thus, the energetics will come out the same.
In particular, if we define
\begin{equation}
\label{GMA1}
\tilde N(x)  = \int_{\mathbb T_{nL}^{d}}\hspace{-4 pt}\tilde \rho_{\star}^{[nL]}(y)
V(x-y)[nL]^{d}dy
\end{equation}
then $\tilde N$ is the periodic extension of the function $N(x)$ given by
\begin{equation}
\label{GMA2}
N(x)  =  \int_{\mathbb T_{L}^{d}}\hspace{4 pt}\rho_{\star}^{[L]}(y)V(x-y) L^{d}dy
\end{equation}
(Note that in Eqs.(\ref{GMA1}) -- (\ref{GMA2}) the factors of $(nL)^{d}$ and $L^{d}$ have been brought inside so that the {\it integrands} are both ostensibly  of order unity and therefore the ``same'' function.)

Thus, again,
\begin{equation}
\label{GKG}
\frac{1}{2}\theta[nL]^{d}\mathcal E^{[nL]}(\tilde \rho_{\star}^{[nL]},\tilde \rho_{\star}^{[nL]})
=
\frac{1}{2}\theta 
n^{d}\times\frac{1}{n^{d}}
L^{d}\mathcal E^{[L]}(\rho_{\star}^{[L]}, \rho_{\star}^{[L]}).
\end{equation}
Altogether, we find 
\begin{equation}
%\label{}
\mathcal F_{\theta}^{[nL]}(\tilde \rho_{\star}^{[nL]}) 
= -\log n^{d} + F_{\theta}^{[L]}(\rho_{\star}^{[L]})
\end{equation}
while, from Eq.(\ref{LA}) with $L_{a} = nL$
\begin{equation}
%\label{}
\mathcal F_{\theta}^{[nL]}(\rho_{0}^{nL})  =  -\log n^{d}  + F_{\theta}^{[L]}(\rho_{0}^{[L]}).
\end{equation}
From the above two equations, we may conclude that $\theta\t^{[L]} \geq \theta\t^{[nL]}$.  

Now consider $L$, $K$ with $K \gg L$ when $K$ is not an integer multiple of $L$.  We will use almost exactly the above argument except that we will acquire an error due to ``boundary terms''. 
Let us find $n$ such that
\begin{equation}
%\label{}
nL < K < (n+1)L
\end{equation}
We shall treat $\mathbb T_{K}^{d}$ like the hypercube $[0,K]^{d}$ which is divided into 
$n^{d}$ hypercubes of scale $L$ which occupy $[0,nL]^{d}$.  For future reference, we refer to boxes 
that share a face with the region $\mathbb T_{K}^{d}\setminus [0,nL]^{d}$ as 
{\it boundary} boxes.  It is noted that there are $B(n,d) = n^{d} - (n-2)^{d}$ such boxes.  
In the region $[0,nL]^{d}$, we define, similar to before, the density 
$\tilde{\rho}_{\star}^{[K]}$ which is the rescaled periodic extension of 
$\rho_{\star}^{[L]}$, the non--trivial density which minimizes
the free energy at $\theta\t^{[L]}$ on $\mathbb T_{L}^{d}$.  In the complimentary region, we set 
$\tilde{\rho}_{\star}^{[K]}$ to zero.  

The entropic calculation proceeds exactly as before with the same result namely
$-\log n^{d}  + \mathcal S^{[L]}(\rho_{\star}^{[L]})$.  However, for the energy integrals, we cannot simply use Eq.(\ref{GKG}) because, e.g.,~if both $x$ and $y$ are in boundary cubes (on opposite sides) the formula may be in error because $V(x-y)$ no longer ``reaches around''.  However for present purposes, it is sufficient to use the results of 
Eq.(\ref{GKG}) and subtract the maximum possible gain from these cubes -- which would be
$-V_{0}$.  The result is
\begin{align*}
%\label{}
K^{d}\mathcal E(\tilde \rho_{\star}^{[K]},\tilde \rho_{\star}^{[K]})  \leq  
\left (
\frac{K}{nL}
\right )^{d}&[L^{d}\mathcal E(\rho_{\star}^{[L]}, \rho_{\star}^{[L]})]
+  \\
&+
V_{0}\cdot [nL]^{d}\int_{{\bf H}(n,d)^{2}}
\hspace{-15 pt}
\tilde \rho_{\star}^{[K]}(x)\tilde \rho_{\star}^{[K]}(y)dxdy]
\end{align*}
where ${\bf H}(n,d)$ is standing notation for the above described region of boundary cubes.  
Note that the $\tilde \rho_{\star}^{[K]}$'s are normalized to $n^{-d}$ in each such cube 
and the order of $B(n,d)$ is $n^{d-1}$.  As a result, 
\begin{equation}
%\label{}
K^{d}\mathcal E(\tilde \rho_{\star}^{[K]},\tilde \rho_{\star}^{[K]})  \leq  
[L^{d}\mathcal E(\rho_{\star}^{[L]}, \rho_{\star}^{[L]})](1 + O(K/L))
\end{equation}
Consequently, for any $\epsilon > 0$, 
\begin{equation}
%\label{}
\theta\t^{[L]} + \epsilon > \theta\t^{[K]}
\end{equation}
for all $K$ sufficiently large which implies the desired result.
\end{proof}

\subsection{Stable Behavior}

Since $\theta\t^{[L]}$ tends to a definitive limit which
(for $V \in \mathscr V_{N}$) is not infinite, it is important to establish the criterion for when this limit is not zero.  As it turns out, the correct condition is closely related to thermodynamic stability.

For two body interactions, the condition of {\it H--stability} (see
\cite{R} p.~34)
is as follows:
$\exists b > -\infty$ such that for any $N$ points,
$x_{1}, \dots x_{N}$
 in $\mathbb R^{d}$, 
\begin{equation}
%\label{}
\sum_{i \neq j}V(x_{i} - x_{j}) \geq -bN.
\end{equation}
This condition is necessary and sufficient for the existence of thermodynamics -- although the early proofs usually assume continuity properties of the interaction.  Provided that $V$ is bounded and continuous, H--stability is equivalent to the condition that for all probability measures described by a density $\rho(x)$,
\begin{equation}
%\label{}
\int_{\mathbb R^{d}\times \mathbb R^{d}}
\hspace{-10 pt}
V(x-y) \rho(x)\rho(y)dxdy
\geq 0
\end{equation}
(c.f. \cite {R}) as is easily seen by utilizing sums of point masses to approximate probability measures.  This will be our working hypotheses for the benefit of the next result along with the technical assumption that $V$ is bounded:

\noindent {\bf Definition} \hspace {2 pt}  An interaction
$V \in \mathscr V$ is said to satisfy {\it condition--K}
if $|V(x)| \leq V_{\text{max}} < \infty$ for 
all $x\in \mathbb R^{d}$ and
if for all $L$ sufficiently large, the inequality
\begin{equation}
%\label{}
\int_{\mathbb T_{L}^{d}\times \mathbb T_{L}^{d}}
\hspace{-10 pt}
V(x-y) \rho(x)\rho(y)dxdy
\geq 0
\end{equation}
holds  for all $\rho \in \mathscr P^{[L]}$.

The principal result of this section is as follows:
\begin{theorem}
Let $V \in \mathscr V_{N}$ denote an interaction that satisfies condition--K and has (for all $L$ sufficiently large) discontinuous lower transitions at 
$\theta\t^{[L]}$ on $\mathbb T_{L}^{d}$.  Then the
$\theta\t^{[L]}$ tends to a limit that is strictly positive.  
\end{theorem}
\begin{proof}

For connivence in the up and coming we shall streamline notation -- e.g., revert to the omission of all $L$'s in the superscripts, etc.  We start by assuming  
$\theta \geq \theta\t$ 
and write, for this value of $\theta$, a non--trivial minimizer and its deviation
$$
\rho  = \rho_{0}(1 + \eta).
$$
Further, we define the positive and negative parts of 
$\eta$ as $\eta^{+}$ and, $\eta^{-}$ respectively and, finally,
$$
h = ||\rho_{0}\eta||_{1}.
$$
The aim is to show that if $\theta$ is small than, regardless of $L$, $h$ must be zero.  

The first step will be an estimate on the free energetics.  We have
\begin{equation}
%\label{}
0 \leq \mathcal F_{\theta}(\rho_{0}) - \mathcal F_{\theta}(\rho)  =  S(\rho_{0}) - S(\rho) 
+ \frac{1}{2}\theta L^{d}
[\mathcal E(\rho_{0},\rho_{0})-\mathcal E(\rho,\rho)].
\end{equation}
As has been stated before, we have
$ \frac{1}{2}\theta L^{d}\mathcal E(\rho_{0},\rho_{0})
  =  \frac{1}{2}\theta v$ while
\begin{equation}
%\label{}
\frac{1}{2}\theta L^{d}\mathcal E(\rho,\rho)  =
\frac{1}{2}\theta v + 
\frac{1}{2}\theta L^{d}
\mathcal E(\rho_{0}\eta,\rho_{0}\eta).
\end{equation}
Let us decompose:
\begin{align*}
%\label{}
\frac{1}{2}\theta L^{d}
\mathcal E(\rho_{0}\eta,\rho_{0}\eta) & =  
\frac{1}{2}\theta L^{d}
[\mathcal E(\rho_{0}\eta^{+},\rho_{0}\eta^{+}) + 
\mathcal E(\rho_{0}\eta^{-},\rho_{0}\eta^{-}) - \\
& 2 \mathcal E(\rho_{0}\eta^{+},\rho_{0}\eta^{-})].
\end{align*}
The first two terms are positive by the hypothesis that 
$V$ satisfies the condition--K thus
\begin{align*}
%\label{}
\frac{1}{2}\theta L^{d}\mathcal E(\rho,\rho)  &\geq 
\frac{1}{2}\theta v - \theta\int_{\mathbb T^{d}_{L}\times \mathbb T^{d}_{L}}
\hspace{- 7 pt}
V(x-y)\eta^{-}(y)\rho_{0}\eta^{+}(x)dxdy
\\
&\geq \frac{1}{2}\theta v - \frac{1}{2}\theta ||V||_{1}h
\end{align*}
where we have used $||\eta^{-}||_{\infty} \leq 1$
and $||\rho_{0}\eta^{+}||_{1}  =  \frac{1}{2}h$.

Putting these together we have
\begin{equation}
\label{6G6}
\mathcal S(\rho) - \mathcal S(\rho_{0}) \leq
\frac{1}{2}\theta L^{d}[\mathcal E(\rho_{0},\rho_{0})
-\mathcal E(\rho,\rho)]
\leq \frac{1}{2}\theta||V||_{1}h.
\end{equation}
Incidentally we may use the {\it lower} bound
(See, \cite{V} p.~271)
$\mathcal S(\rho) - \mathcal S(\rho_{0}) \geq
\frac{1}{2}h^{2}$ to learn that the assumption that 
$\theta$ is ``small'' necessarily implies that
$h$ is small but the particulars of this bound does not play a major r\^{o}le.

Now, let us write the mean--field equation, Eq.(\ref{VEV}),
in a form useful for the present purposes:
\begin{equation}
\label{Aaa}
\log \rho + \theta L^{d}
\hspace{-2 pt}
\int_{\mathbb T_{L}^{d}} 
\hspace{-4 pt}
V(x-y) \rho(y)dy
=  C_{\text{KM}}
\end{equation}
with $C_{\text{KM}}$ a constant that we are now prepared 
to ``evaluate''
\begin{equation}
\label{Aab}
C_{\text{KM}}  =  \mathcal S(\rho) + 
\theta L^{d}\mathcal E(\rho,\rho).
\end{equation}
Expressing Eqs.(\ref{Aaa}) -- (\ref{Aab}) for the benefit 
of $\eta$ we have
\begin{equation}
%\label{}
\log \rho_{0} + \log(1 + \eta)
+ \theta v +
\theta
\hspace{-2 pt}
\int_{\mathbb T_{L}^{d}} 
\hspace{-4 pt}
V(x-y) \eta(y)dy
=
\mathcal S(\rho) + \theta L^{d}\mathcal E(\rho,\rho)
\end{equation}
so
\begin{align*}
%\label{}
\log(1 + \eta) + 
\theta
\hspace{-2 pt}
\int_{\mathbb T_{L}^{d}} 
\hspace{-4 pt}
V(x-y) \eta(y)dy
&=: -\kappa =  \\
& \mathcal F_{\theta}(\rho) - 
\mathcal F_{\theta}(\rho_{0}) +
\frac{1}{2}\theta L^{d}
[\mathcal E(\rho,\rho) - \mathcal E(\rho_{0},\rho_{0})]
\end{align*}
where it is noted that the sign of $\kappa$ is pertinent. 
Indeed by the display just prior to Eq.(\ref{6G6}) we have
\begin{equation}
%\label{}
0 \leq \kappa \leq \frac{1}{2}\theta||V||_{1}h
\end{equation}
(Thus, in addition, $\kappa$ is small).  Note that all the above holds formally even if $\eta = 0$ so in the future, we need not insert provisos. 

For use in the remainder of this proof, we shall divide
$\mathbb T_{L}^{d}$ into disjoint cubes $C_{1}, \dots
C_{j}, \dots $
(half open/closed etc.) of diameter 
$a$.  Since we are only pursuing 
the limit of $\theta\t^{[L]}$, it may just as well be assumed that 
$2a$ divides $L$.  We use the notation
$$
||f||_{L^{1}(C_{j})}  :=  \int_{C_{j}}
\hspace{-1 pt}
|f(x)|dx
$$
and similarly for other local norms.  

Our first substantive 
claim is as follows:  

\noindent  Let $\epsilon$ denote a small number 
which is of order unity independent of $L$ 
(the nature of which is not so important and will, to some extent, be clarified below) and suppose that for all $j$, 
\begin{equation}
%\label{}
||\eta||_{L^{1}(C_{j})} \leq \frac{\epsilon}{\theta}.  
\end{equation}
Then, for all $\theta$ is sufficiently small, for all
$L$ under discussion, we have that
 $\eta \equiv 0$.  

To see this we let
$x \in C_{j}$ be in the support of 
$\eta^{+}$ so, ostensibly, we have
\begin{equation}
%\label{}
1 + \eta^{+}(x)  =  \text{e}^{-\kappa}
\text{e}^{-\theta\int_{\mathbb T_{L}^{d}} 
\hspace{-2 pt}
V(x-y)\eta(y)dy}.
\end{equation}
However, due to the finite range of $V$, the integration actually takes place
on only the cubes in the immediate vicinity of 
$C_{j}$ so that
\begin{equation}
%\label{}
1 + \eta^{+}(x)  \leq  \text{exp}[\theta V_{\text{max}}
\sum_{j^{\prime} \sim j}||\eta||_{L^{1}(C_{j^{\prime}})}]
\leq  \text{e}^{\epsilon V_{\text{max}}D_{1}}
\end{equation}
where $j^{\prime} \sim j$ means that
$\overline{C}_{j^{\prime}}\cap \overline{C}_{j} \neq \emptyset$
and $D_{1} = D_{1}(d)$ is the number of 
$j^{\prime}$ such that $j^{\prime} \sim j$.  
This implies an $L^{\infty}$--bound on 
$\eta^{+}$ which is (a small number) of
order unity.  We run a similar argument for
$\eta^{-}$ -- only now we have to contend with 
$\kappa$:
\begin{equation}
%\label{}
1 - \eta^{-}  \geq  \text{e}^{-\kappa}
\text{e}^{-\epsilon V_{\text{max}}D_{1}}
\end{equation}
i.e., 
\begin{equation}
%\label{}
\eta^{-}  \leq  \kappa + \epsilon V_{\text{max}}D_{1}.
\end{equation}
Thus we have an $L^{\infty}$ bound on the full
$\eta$ which implies, at this stage -- since $h$ and
$\theta$ are supposed to be small --
an improved bound on 
$||\eta||_{L^{1}(C_{j})}$ in all cubes 
$C_{j}$.  Let us continue the process.  Suppose that 
for all $j$
\begin{equation}
%\label{}
||\eta||_{L^{1}(C_{j})}  \leq  \phi
\end{equation}
where 
$\phi = \phi(\theta,h)$ represents the latest improvement.
Then
\begin{equation}
%\label{}
1 + ||\eta^{+}||_{\infty}
\leq \exp [ D_{1}V_{\text{max}}\phi \theta]
\end{equation}
-- so that 
$||\eta^{+}||_{\infty}
\lesssim  D_{1}V_{\text{max}}\phi \theta$
-- and
\begin{equation}
%\label{}
||\eta||_{\infty} \leq D_{1}V_{\text{max}}\phi \theta
+ \kappa
\end{equation}
which, at least for a while, represents an improvement on the various $L^{1}$--norms.

The procedure is no longer beneficial if $\phi$ 
is on the order of $\kappa$.  E.g. we may stop when
$\kappa \geq \phi
[[2a]^{-d} - \theta V_{\text{max}}D_{1}]$
-- where it is assumed that $\theta$
is small enough so that the coefficient of 
$\phi$
is positive.  
We arrive at an overall bound:
$$
||\eta||_{\infty} \leq
c_{a}\kappa \leq c_{b}\theta h
$$
where the $c$'s are constants of order unity independent of $L$ and $\theta$ (provided that the latter is sufficiently small).
Clearly, for $\theta$ small enough, 
this cannot be consistent with
$||\eta\rho_{0}||_{1}  =  h$ unless $h = 0$.

Thus, we are done with the proof unless there are bad blocks where the local $L^{1}$--norm 
of $\eta$ is in excess of 
$\epsilon \theta^{-1}$.  In fact, we will present an additional
hierarchy of bad blocks.  The above blocks will be the {\it core blocks} which will be denoted by
${\bf C}$.   We shall define blocks 
${\bf B}_{n}$, $n = 0, 1, \dots s$ by
\begin{equation}
%\label{}
{\bf B}_{n}  =  
\{C_{j} \mid \epsilon \theta^{n} < ||\eta||_{L^{1}(C_{j})}
\leq \theta^{n}\}
\end{equation}
and (unfortunately) other bad blocks
\begin{equation}
%\label{}
{\bf B}^{\prime}_{n}  =  
\{C_{j} \mid \epsilon \theta^{n-1} \geq ||\eta||_{L^{1}(C_{j})}
> \theta^{n}\}.
\end{equation}
Thus it is seen that our 
$\epsilon$ should be small enough to absorb various constants which crop up but large compared to the 
assumed value of $\theta$.
The hierarchy of these sets stops when the 
$L^{1}$--norm is comparable to $\kappa$ 
-- pretty much as in the previous argument.  Here we shall say that $s$ is defined so that
blocks outside the hierarchy have local $L^{1}$--norm
of $\eta$ less than a constant $Q_{0}$ times 
$\kappa$ with $Q_{0}$ to be described shortly.  
  Such blocks will be informally referred to as {\it background} blocks.  

We order the hierarchy in the obvious fashion:
\begin{equation}
%\label{}
\dots \ {\bf B}_{n}^{\prime} \succ {\bf B}_{n} \succ {\bf B}_{n+1}^{\prime} \succ \dots 
\end{equation}
with every element of the hierarchy considered to be above the background blocks and below the core blocks.  

Our next claim is that for any block in the hierarchy there must be a neighboring block which is further up the hierarchy.  
%(This statement does not include the core blocks which may be considered ``above'' the hierarchy.)
Indeed suppose not and that, e.g.,
$C_{j} \in B_{n}$.  If all neighbors -- that is to say all blocks $C_{j^{\prime}}$ with $j^{\prime} \sim j$ -- were only at the level ${\bf B}_{n}$ or below, we would have, for $x \in C_{j}$
\begin{equation}
%\label{}
1 + \eta^{+}(x)  \leq  
\text{e}^{\theta D_{1} V_{\text{max}}\theta^{n}}
\end{equation}
and 
\begin{equation}
%\label{}
\eta^{-}(x) \leq \theta D_{1} V_{\text{max}}\theta^{n}
+ \kappa.
\end{equation}
This, for most cases, implies that 
$||\eta||_{L^{\infty}(C_{j})}$ is of the order 
$\theta^{n+1}$ which precludes 
$C_{j} \in {\bf B}_{n}$; a similar derivation applies to the 
${\bf B}_{n}^{\prime}$.  
At the very bottom of the
hierarchy, the same situation holds with any reasonable choice of $Q_{0}$ which is of the order of unity.  

The implication of the preceding claim is that each 
block in the hierarchy is connected to the core by a path  whose length does not exceed the order of its the hierarchal index.

Our next claim is that (under the hypothesis of non--triviality) the vast majority of the $L^{1}$--norm of 
$\eta$ is carried by the core and its immediate vicinity.  First, let us estimate 
$|{\bf C}|$, the volume of the core.  We argue that 
\begin{equation}
\label{8U8}
|{\bf C}|\cdot\epsilon\cdot\frac{1}{\theta}  \leq  
hL^{d}
\end{equation}
since the right side is the full $L^{1}$--norm of 
$\eta$ and the left side represents the minimal contribution to this effort on the part of the core.  
Thus the core volume fraction is the order of
$\theta h$ which, we remind the reader is supposed small.

Now, by the connectivity property of the hierarchy, Eq.(\ref{8U8}) can be translated into an estimate on the volume of the sets 
${\bf B}_{n}$, ${\bf B}_{n}^{\prime}$.  Indeed, we may write
\begin{equation}
%\label{}
|{\bf B}_{n}|  \leq  |{\bf C}|D_{2}n^{d}
\end{equation}
where $D_{2}  =  D_{2}(d)$ is another constant.  

The above two estimates are sufficient to vindicate the claim at the beginning of the paragraph containing Eq.(\ref{8U8}).  We denote by 
${\bf C}^{\star} = {\bf C}\cup {\bf B}_{0}\cup {\bf B}_{0}^{\prime}$ which we call the extended core. Turning attention to the complimentary set we have, for ${\bf B}_{n}$:
\begin{equation}
%\label{}
\int_{{\bf B}_{n}}
\hspace{-5 pt}
|\eta|dx
\leq  \theta^{n}|{\bf C}|D_{2}n^{d}[2a]^{-d} 
\leq h\theta L^{d}\frac{1}{\epsilon}
\cdot \theta^{n}n^{d}D_{2}[2a]^{-d}
\end{equation}  
Summing this over {\it all} $n$ starting from $n = 1$
we get a contribution (for all $\theta$ sufficiently small) which is bounded by 
$K \epsilon^{-1}\theta^{2}hL^{d}$ for some constant $K$.  
Similarly 
\begin{equation}
%\label{}
\int_{{\bf B}_{n}^{\prime}}
\hspace{-5 pt}
|\eta|dx
\leq  
\epsilon\theta^{n-1}|{\bf C}|D_{2}n^{d}[2a]^{-d}
\end{equation}
whence the total contribution from the primed portion of the hierarchy on the compliment of the extended core to $||\eta||_{1}$ is no more than
$K^{\prime} \theta hL^{d}$ for some constant $K^{\prime}$  Note that these are small compared to the purported total of $h L^{d}$.  Finally, the contribution from all the background blocks is surely no more than
$Q_{0}[2a]^{-d}\kappa [L/a]^{d}$ 
and we recall that $\kappa$ itself is bounded above by the order of $\theta h$.
So, in summary, we arrive at
\begin{equation}
%\label{}
\int_{{\bf C}^{\star}}\eta^{+}dx \geq  
(1 - g\theta)hL^{d}
\end{equation}
for some constant $g$ which is independent of $L$
and $\theta$ (for $\theta$ sufficiently small).  

With the preceding constraint in mind, let us bound from {\it below} the relative entropy  
\begin{equation}
%\label{}
\mathcal S(\rho) - \mathcal S(\rho_{0})
=
\int_{\mathbb T_{L}^{d}}
\hspace{-3 pt}
\rho_{0}(1 + \eta^{+})\log(1 + \eta^{+})
+
\int_{\mathbb T_{L}^{d}}
\hspace{-3 pt}
\rho_{0}(1 - \eta^{-})\log(1 - \eta^{-}).
\end{equation}
The second term may be bounded below by 
$-\frac{1}{2}h$.  As for the former, since the function is always positive, we may restrict attention to the set 
${\bf C}^{\star}$.  As is not hard to show, the contribution from
${\bf C}^{\star}$ is larger than that of the function 
which is uniform on ${\bf C}^{\star}$ and has the same total mass.  As a result:
\begin{equation}
%\label{}
\mathcal S(\rho) - \mathcal S(\rho_{0})
\geq -\frac{1}{2}h +
\frac{|{\bf C}^{\star}|}{L^{d}}
(1 + \frac{1}{|{\bf C}^{\star}|}(1-g\theta)hL^{d})
\log (1 + \frac{1}{|{\bf C}^{\star}|}(1-g\theta)hL^{d}).
\end{equation}
As is not hard to see, (and is intuitively clear) this is decreasing in $|{\bf C}^{\star}|$ -- meaning we may substitute the upper bound based on Eq.(\ref{8U8}):
$|{\bf C}^{\star}| \leq L^{d}G\theta h$
with $G$ ($\propto \epsilon^{-1}$) another constant of order unity independent of $L$ and 
$\theta$.  All in all,
\begin{equation}
\label{AZQ}
\mathcal S(\rho) - \mathcal S(\rho_{0})
\geq -\frac{1}{2}h +
G\theta h
(1 + \frac{(1-g\theta)}{G}\frac{1}{\theta})
\log
(1 + \frac{(1-g\theta)}{G}\frac{1}{\theta}).
\end{equation}

By contrast we have, from Eq.(\ref{6G6}),
that $\mathcal S(\rho) - \mathcal S(\rho_{0})$
is {\it less} than a constant times 
$\theta h$.  This along with 
Eq.(\ref{AZQ}) implies that $h = 0$ or, assuming that $h \neq 0$ 
a strict lower bound on $\theta$.  
Either of these conclusions allows us to infer the desired result.
\end{proof}  

% http://picasaweb.google.com/evelyman/NewAlbum62209635PM?authkey=Gv1sRgCM7SoKWD_t6zfQ&feat=directlink
%\vspace{.125 cm}
\subsection{Catastrophic Behavior}
We conclude with some examples of what can go ``wrong'' if the criterion of thermodynamic stability is violated.  For the benefit of these final results we shall violate condition--K by assuming the existence of a
compactly supported probability density
$\rho_{\odot}(x)$ such that 
\begin{equation}
\label{hCh}
\int_{\mathbb R^{d}\times \mathbb R^{d}}
\hspace{-12 pt}
V(x-y)\rho_{\odot}(x)\rho_{\odot}(y)dxdy = -u_{0}
\end{equation}
with $u_{0} > 0$.
First, some specific instances:
%
%\noindent {\bf Definition} \hspace {2 pt}
%Let $V \in \mathscr V_{N}$  Then $V$ is said to be of type--J
%if $\exists a_{0} > 0$ and a $\rho_{\odot} \in \mathscr P$ (for some $L$) which is supported on the ball
%$|x| \leq a_{0}$ such that
%$$
%\int V(x-y)\rho_{\odot}(x)\rho_{\odot}(y)dxdy = -u_{0}
%$$
%with $u_{0} > 0$.

%\noindent {\bf Reamrk} \hspace {2 pt}  It is remarked that a type-J interaction implies a violation of H--stability.  Indeed, in the present context, H--stability implies that for all  $\rho$, 
%$$
%\int_{\mathbb T_{L}^{d}} V(x-y)\rho(x)\rho(y)dxdy \geq 0.
%$$
{\begin{proposition}
\label{TPZ}
If $V$ satisfies either of the following then it violates condition--K
\\
(a) The interaction $V$ satisfies $\int_{\mathbb R^{d}}V(x)dx = -v_{0} < 0$\\
(b) For some $\lambda < 1$, in a $\lambda a_{0}$ neighborhood of the origin, 
$V$ integrates to $+c_{0}$ while for $\lambda a_{0} \leq |x| \leq 2a_{0}$, $V(x)$ is bounded above by $-v_{0}$ where
$$
v_{0}(1-\lambda^{d}) > c_{0}.
$$
\end{proposition}
\begin{proof}
In the first case, we choose
$$
\rho_{\ell,\odot}(x)  =  \chi_{|x| \leq \ell a} \frac{1}{\gamma[\ell a]^{d}}
$$
where $\gamma$ is a geometric constant.  It is noted that
\begin{equation}
%\label{}
\lim_{\ell \to \infty} g(a\ell)^{d}
\int_{\mathbb R^{d}}V(x-y)\rho_{\ell,\odot}(x)\rho_{\ell,\odot}(y)dxdy  =  -v_{0}
\end{equation}
so the result follows for $\ell$ sufficiently large.

As for the second case, we simply use
$\rho_{\odot}(x)  =  \frac{1}{\gamma[ a_{0}]^{d}} \chi_{|x| \leq a_{0}}$.  
In performing the integration
$$
\int_{\mathbb R^{d}\times \mathbb R^{d}}V(x-y)\rho_{\odot}(x)\rho_{\odot}(y)dxdy
$$
{\it and} ignoring the positive contribution from $|x-y| < \lambda a_{0}$
the result would be not more than $-v_{0}$.  For each $x$ we must cut out a ball of radius (no more than) $\lambda a_{0}$ around each point of the integration and insert a corresponding factor of (no more than) $c_{0}$.  The result
is no more than $-(v_{0}(1-\lambda^{d}) - c_{0})$
and we are done.
It is noted that this latter result applies immediately to the case where $V$ is negative in a deleted neighborhood of the origin.
\end{proof}
\begin{theorem}
For a potentials that violate condition--K via Eq.(\ref{hCh}), the V-McK system exhibits non--physical scaling.  In particular, for $L$ sufficiently large, 
$$
\theta\t(L)  \leq  r L^{-d}
$$
for some $r > 0$.
\end{theorem}
\begin{proof}
Recalling Eq.(\ref{LA}) we have for any $L$
$$
\mathcal F^{[L]}_{\theta}(\rho^{[L]}_{0})  =  -\log L^{d}  +  \frac{1}{2}\theta v.
$$
By contrast, if we abide by the recommended density we obtain:
$$
\mathcal F^{[L]}_{\theta}(\rho_{\odot})  =  \mathcal S(\rho_{\odot}) - 
\frac{1}{2}\theta\rho_{0}^{-1}u_{0}
$$
where it is noted that $\mathcal S(\rho_{\odot}) < \infty$ by hypothesis and by the restrictive nature of the support of $\rho_{\odot}$, it is independent of $L$.  But then, as soon as 
$-\frac{1}{2}\theta u_{0}L^{d} + \mathcal S(\rho_{\odot}) 
< \log L^{d} + \frac{1}{2}\theta g_{0}$
it must be that $\theta \leq \theta\t$.  This implies the stated bound.  
\end{proof}

%\bibitem{MA}
%Nicolas Martzel and Claude Aslangul,
%\textit{Mean-field treatment of the many-body 
%Fokker--Planck equation},
%\textrm{J.Phys A Math. Gen},
%\textbf{34} 11225 --  11240
%(2001).

%1. Tamura, Y. ;  On asymptotic behaviors of the solution of a
%nonlinear diffusion equation, J. Fac. Sci. Univ. Tokyo,
%Sect IA, Math. 31 (1984), pp. 195--221.

%5. Sznitman, A.-S., Topics in propagation of chaos.
%\'Ecole d'\'Et\'e de Probabilit\'es de Saint-Flour XIX---1989,
%pp. 165--251, Lecture Notes in Math., 1464, Springer, Berlin, 1991.

%Catholic Univ., Wash. DC, 1967, pp. 41--57.

%3. McKean, H. P., Jr. A class of Markov processes associated
%with nonlinear parabolic equations.  Proc. Nat. Acad. Sci. U.S.A.
%56 (1966) pp. 1907--1911.


\begin{thebibliography}{a}

\bibitem{aBjCtL} 
A. L. Bertozzi, J. A. Carrillo, and T. Laurent, 
\textit{Blowup in Multidimensional Aggregation Equations with Mildly Singular Interaction Kernels},
\textrm{Nonlinearity}, 
{\bf 22}, 683 -- 710 (2009)

\bibitem{aBtL} 
A. L. Bertozzi and T. Laurent, 
\textit{The Behavior of Solutions of Multidimensional Aggregation Equations with Mildly Singular Interaction Kernels}, to appear in
\textrm{Chinese Ann. Math.}

\bibitem{BGM}
F. Bolley, A. Guillin, F. Malrieu
\textit{Trend to Equilibrium and Particle Approximation for a Weakly Self--Consistent Vlasov-Fokker-Planck Equation},
\textrm{Preprint}, AeXiv.org:0906.1417
(2009)


\bibitem{BL}
P. Butt\`{a} and J.L. Lebowitz,
\textit{Local Mean Field Models of Uniform to Nonuniform Density Fluid-Crystal Transitions},
\textrm{J. Phys. Chem. B,} 
{\bf 109},
6849 -- 6854 (2005).

\bibitem{BL2}
P. Butt\`{a} and J.L. Lebowitz,
\textit{Hydrodynamic Limit of Brownian Particles Interacting with Short and Long Range Forces},
\textrm{J. Stat. Phys.} 
{\bf 94} No 3/4,
653 -- 694 (1999).



\bibitem{CCELM}
E. A. Carlen, M. C. Carvalho, R. Esposito, J. L. Lebowitz, R. Marra,
\textit{Free Energy Minimizers for a Two--Species Model with Segregation and Liquid-Vapor Transition},
\textrm{Nonlinearity},
\textbf{16}, 1075 -- 1105
(2003).


\bibitem{CDP}  J. A. Carrillo, M. R. D'Orsogna 
and V. Panferov, \textit{Double Milling in Self-Propelled 
Swarms from Kinetic Theory},
\textrm{Kinetic and Related Models}, 
{\bf 2}, No 2, 363 -- 378 (2009). 

\bibitem{CDMBC}(MR2369988)
Y.-L. Chuang, M. R. D'Orsogna, D. Marthaler, A. L. Bertozzi and L.
Chayes, {\em State Transitions and the Continuum Limit for a 2D
Interacting, Self-Propelled Particle System}, Physica D,
\textbf{232},  33 -- 47, (2007).

\bibitem{Co}
P. Constantin,
\textit{The Onsager Equation for Corpora,}
\textrm{Journal of Computational
and Theoretical Nanoscience},
to appear
(2009).

\bibitem{DCBC}
M. R. D'Orsogna, Y.-L. Chuang, A. L. Bertozzi and L. S. Chayes, {\em
Self--Propelled Particles with Soft--Core Interactions: Patterns,
Stability, and Collapse}, 
Phys. Rev. Lett., \textbf{96}, (2006),
104302~1-4.

\bibitem{EGM}
R. Esposito, Y. Guo and R. Marra, 
\textit{Phase Transition in a
Vlasov-Boltzmann Binary Mixture,}
ArXiv 0904.0791v1 [math-ph] (2009).

\bibitem{Ga}
D.J. Gates,
\textit{Rigorous Results in the Mean--Feld Theory of Freezing},
\textrm{Ann. Phys},
\textbf{71}, 395 -- 420
(1972).

\bibitem{GP}
D.J. Gates and O. Penrose,
\textit{The van der Waals Limit for Classical Systems
III.  Deviation from the van der Waals--Maxwell Thoery},
\textrm{Commun. Math. Phys},
\textbf{17}, 194 -- 209
(1970).


\bibitem{GK}
N. Grewe and W. Klein,
\textit{The Kirkwood -- Salsburg Equations for a Bounded Stable Kac Potential.  II.  Instability and Phase Transitions}
\textrm{J. Math. Phys}
\textbf{18} No.~9,
1735 -- 1740
(1977)

\bibitem{HS}
J. Haskovec, C. Schmeiser,  \textit{Stochastic Particle Approximation for Measure Valued Solutions of the 2D Keller-Segel System}, \textrm{J. Stat. Phys.}, 
{\bf 135} No 1, 133 -- 151 (2009)

\bibitem{Kac}
M. Kac,
\textit{On the Partition Function of a One-Dimensional Gas},
\textrm{Phys. Fluids.},
\textbf{87}, 8 -- 12
(1959)


\bibitem{KUH}
M. Kac, G.E. Uhlenbeck and P.C. Hemmer
\textit{On the van der Waals Theory of the Vapor--Liquid Equilibrium. I.  Discussion of a One--Dimensional Model},
\textrm{J. Math. Phys},
\textbf{2} No 1, 216 -- 228
(1963)

\bibitem{vKa}
N.G. van Kampen,
\textit{Condensation of a Classical Gas with Long-Range Attraction},
\textrm{Phys Rev},
\textbf{135}, A362 -- A369
(1964)

\bibitem{KS} E. F. Keller, L. A. Segel,  
\textit{Initiation of Slime Mold Aggregation Viewed as 
an Instability},  
\textrm{J. Theoret. Biol.} 
\textbf{26}, 399 -- 415  (1970)

\bibitem{KM}
J. G. Kirkwood and E. J. Monroe,
\textit{Statistical Mechanics of Fusion},
\textrm{J. Chem. Phys.},
\textbf{9}, 514 -- 526
(1941)

\bibitem{LP}
J. Lebowitz and O. Penrose,
\textit{Rigorous Treatment of the van der Waals - Maxwell Theory of the Liquid--Vapour Transition},
\textrm{J. Math. Phys.},
\textbf{7}, 98 -- 113
(1966)


\bibitem{LR}
H. Levine, W.J. Rappel and I. Cohen, {\em Self-Organization in
Systems of Self-Propelled Particles}, Phys. Rev. E, \textbf{63},
(2000), 017101-1/4.

\bibitem{Mal}
F. Malrieu,
\textit{Logarithmic Sobolev Inequalities for some Nonlinear PDE's }
\textrm{Stochastic Process. Appl.} 
{\bf 95} No 1,
109 -- 132 (2001).

\bibitem{MaMa}
G. Manzi and R. Marra
\textit{Phase Segregation and Interface Dynamics in Kinetic Systems},
\textrm{Nonlinearity},
\textbf{19}, 115 -- 147
(2006)

\bibitem{MA}
N. Martzel and C. Aslangul,
\textit{Mean-field treatment of the many-body 
Fokker--Planck equation},
\textrm{J.Phys A Math. Gen},
\textbf{34}, 11225 --  11240
(2001).

\bibitem{McK}
H. P. McKean Jr.,
\textit{Propagation of Chaos for a Class of
Non-Linear Parabolic Equations}, in \textit{Stochastic Differential
Equations},
\textrm{Lecture Series in Differential Equations},
\textbf{7}, 41 -- 57
(1967)

\bibitem{R}
D. Ruelle, 
\textit{Statistical Mechanics, Rigorous Results,} 
New York: W.A. Benjamin (1969).

\bibitem{JT2}
S. Shigeo and Y. Tamura,
\textit{Gibbs Measures for Mean Field Potentials}
J. Fac. Sci. Univ. Tokyo,
Sect IA, Math. 
\textbf{31},
223 -- 245 (1984).

\bibitem{aS}
A. Stevens
\textit{The Derivation of Chemotaxis Equations as Limit Dynamics of Moderately Interacting Stochastic Many-Particle Systems},
\textrm{SIAM J. Appl. Math.},
\textbf{61} No 2, 183 -- 212
(2000)

\bibitem{Szn}
A.S.~Sznitman,
\textit{Topics in Propagation of Chaos},
\'Ecole d'\'Et\'e de Probabilit\'es de Saint-Flour,
Lecture Notes in Math., 1464
165 -- 251, Springer, Berlin
 (1991)

\bibitem{JT}
Y. Tamura,
\textit{On Asymptotic Behaviors of the Solution of a
Nonlinear Diffusion Equation},
J. Fac. Sci. Univ. Tokyo,
Sect IA, Math. 
\textbf{31},
195 -- 221 (1984).

\bibitem{CjT}
C. J. Thompson,
\textit{Validity of Mean-Field Theories in Critical Phenomena},
\textrm{Prog. Theo. Phys.},
\textbf{87}, 535 -- 559
(1992)



\bibitem{V}
C. Villani, 
\textit{Topics in Optimal Transportation,} Graduate Studies
in Mathematics, Volume 58, AMS (2003).

%\noindent\textcolor{magenta}{\bibitem{VH}
%C. Villani
%\textit{Hypoccc}}

%\noindent\textcolor{magenta}{\bibitem{BGM}
%F. Bolley, etc.}

\end{thebibliography}
\end{document}